\documentclass[conference]{IEEEtran}
\IEEEoverridecommandlockouts

\usepackage{amssymb}
\usepackage{amsfonts}
\usepackage{amsmath}
\usepackage{amsthm}
\usepackage{float}
\usepackage{epsfig}
\usepackage{graphicx}
\usepackage{subcaption}
\usepackage{tabularx}
\usepackage{bm}
\usepackage{dsfont}

\usepackage{url}
\usepackage{amssymb}
\usepackage{amsfonts}
\usepackage{amsmath}
\usepackage{amsmath,amssymb,verbatim}
\usepackage{graphicx}
\usepackage{subcaption}
\usepackage{float}
\usepackage{epsfig}
\usepackage{threeparttable}
\usepackage{algorithm}
\usepackage{algorithmic}
\usepackage{color} 
\usepackage{dblfloatfix}
\usepackage{changepage}
\usepackage{comment}
\usepackage{multirow}
\usepackage{colortbl}
\usepackage{ulem}
\usepackage{color}
\usepackage{lipsum}
\usepackage{cite}
\usepackage{bbm}
\usepackage{comment}
\usepackage{ulem}

\usepackage{tabularx}

\def\BibTeX{{\rm B\kern-.05em{\sc i\kern-.025em b}\kern-.08em T\kern-.1667em\lower.7ex\hbox{E}\kern-.125emX}}

\floatname{algorithm}{Algorithm}

\newtheorem{theorem}{\it Theorem}

\newtheorem{lemma}{\it Lemma}
\newtheorem{corollary}{\it Corollary}
\newtheorem{definition}{\it Definition}

\newtheorem{example}{\it Example}

\definecolor{darkgreen}{rgb}{0.0, 0.5, 0.0}

\begin{document}

\title{Towards Optimal Tradeoff Between Data Freshness and Update Cost in Information-update Systems}

\author{Zhongdong Liu, Bin Li, Zizhan Zheng, Y. Thomas Hou, and Bo Ji
\thanks{The work of Bo Ji and Zhongdong Liu was supported in part by the NSF under Grants CNS-2112694 and CNS-2106427.
The work of Zizhan Zheng was supported in part by the NSF under Grant CNS-1816943. 
The work of Y.~T.~Hou was supported in part by ONR MURI grant N00014-19-1-2621, Virginia Commonwealth Cyber Initiative (CCI), and Virginia Tech Institute for Critical Technology and Applied Science (ICTAS). 
A preliminary version of this work was
presented at ICCCN 2022 as an invited paper \cite{zdicccn22}.
}
\thanks{Zhongdong Liu (zhongdong@vt.edu) and Bo Ji (boji@vt.edu) are with the Department of Computer Science, Virginia Tech, Blacksburg, VA. 
Bin Li (binli@psu.edu) is with the Department of Electrical Engineering, the Pennsylvania State University, State College, PA. 
Zizhan Zheng (zzheng3@tulane.edu) is with the Department of Computer Science, Tulane University, New Orleans, LA. 
Y.~Thomas Hou (thou@vt.edu) is with the Bradley Department of Electrical and Computer Engineering, Virginia Tech, Blacksburg, VA.
}}

\maketitle

\begin{abstract}
In this paper, we consider a discrete-time information-update system, where a service provider can proactively retrieve information from the information source to update its data and users query the data at the service provider. One example is crowdsensing-based applications. In order to keep users satisfied, the application desires to provide users with fresh data, where the freshness is measured by the \emph{Age-of-Information (AoI)}. However, maintaining fresh data requires the application to update its database frequently, which incurs an update cost (e.g., incentive payment). Hence, there exists a natural tradeoff between the AoI and the update cost at the service provider who needs to make update decisions. 
To capture this tradeoff, we formulate an optimization problem with the objective of minimizing the total cost, which is the sum of  the staleness cost (which is a function of the AoI) and the update cost. Then, we provide two useful guidelines for the design of efficient update policies. 
Following these guidelines and assuming that the aggregated request arrival process is Bernoulli, we prove that there exists a threshold-based policy that is optimal among all online policies and thus focus on the class of threshold-based policies. Furthermore, we derive the closed-form formula for computing the long-term average cost under any threshold-based policy and obtain the optimal threshold. 
Finally, we perform extensive simulations using both synthetic data and real traces to verify our theoretical results and demonstrate the superior performance of the optimal threshold-based policy compared with several baseline policies.
\end{abstract}


\section{Introduction}
With the remarkable development of communication networks and smart portable devices in recent years, we have witnessed significant advances in crowdsensing-based applications (e.g., Google Waze \cite{Waze} and GasBuddy \cite{GasBuddy}).  These applications provide services to users by resorting to the community to sense and send back real-time information (e.g., traffic conditions and gas prices) \cite{li2021achieving}. 
To satisfy the diverse needs of users,  such applications need to maintain their knowledge of a set of distributed points of interest (PoI).
For example, GasBuddy monitors gasoline prices at a large number of scattered gas stations in a certain area. 
In order to quickly and accurately respond to users’  requests, the applications need to keep their data fresh. However, given the dynamic changes of the data, maintaining the freshness of data introduces a natural tradeoff between data freshness and update cost. On the one hand, users are unsatisfied if the responses to their requests are outdated; on the other hand, there is a cost for the applications to update their data because updating data relies on user feedback and often requires monetary payment to incentivize users \cite{GasBuddy,li2021achieving}.

\begin{figure}[!t]
\centering
\includegraphics[width=0.42\textwidth]{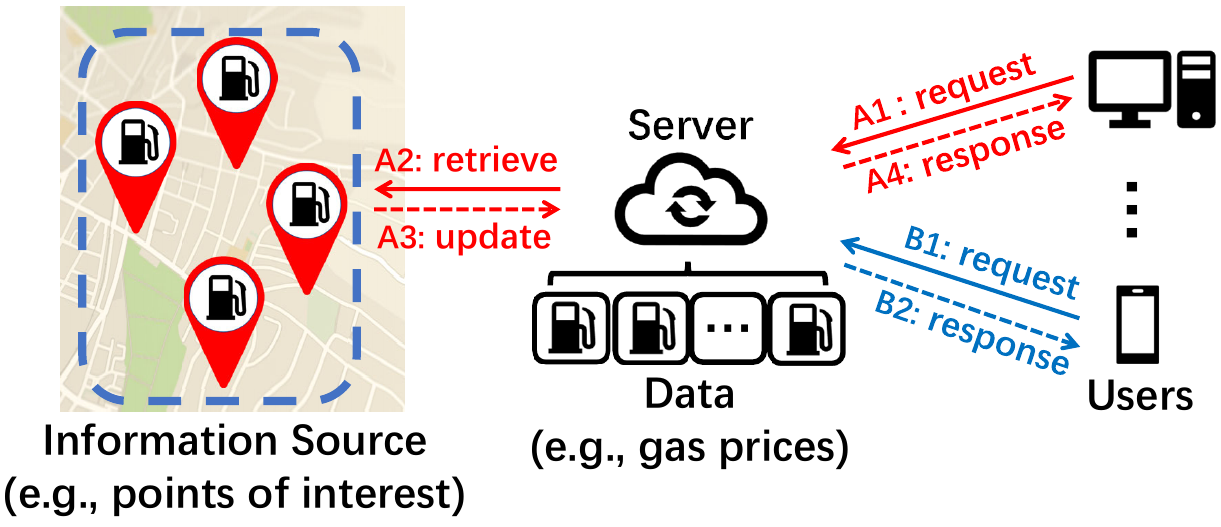}
\caption{An illustration of our system model. Upon receiving a request from the users, the server can either first update the data and then reply (red path: A$1$-A$4$) or simply reply with local data (blue path: B$1$-B$2$).}
\label{fig:model}
\end{figure}

In fact, the tradeoff between data freshness and update cost does not only exist in crowdsensing-based applications, but also in a wide variety of time-sensitive data-driven applications that require timely information updates~\cite{kaul12infocom,li2020waiting, zhongdong, ling2004tradeoff, MarketData}.
For example, in stock analysis applications, the server keeps track of the prices of a large number of stocks and generates different versions of the analysis reports for the stocks at certain times, and users send requests to query these analysis reports~\cite{MarketData}.
To ensure that users receive the real-time analysis of the stock they are trading with, the server needs to retrieve timely information (e.g., various stock market indexes) from the stock market,
which incurs an update cost (e.g., bandwidth resources).
Similar applications also include news feeds, weather updates,
and flight aggregators.

The aforementioned applications have two notable characteristics: 
First,  the server can proactively retrieve information from the information source to update its data, and users need to query the server to obtain the data (i.e., the ``Pull'' model~\cite{sangglobecom17,li2020waiting});
Second, the responses to users' requests (e.g., gas prices) typically do not require significant processing, and the packet size is usually small, making the packet transmission time negligible. However, retrieving the data from the information source often requires certain resources and introduces costs. 
These two characteristics not only distinguish such applications from other ones whose update costs mainly come from service and communication delays \cite{huang2015optimizing, chen2016age, kadota2016minimizing} but also lead to the tradeoff between the data freshness and the update cost.

To that end, in this work, we aim to optimize the tradeoff between data freshness and update cost.
Specifically, we consider a discrete-time system in the setting where a service provider can proactively retrieve information from the information source and users obtain the data at the service provider by sending requests (see Fig.~\ref{fig:model}).
The freshness of the data received by users is measured by a popular timeliness metric called \textit{Age-of-Information (AoI)} \cite{kaul12infocom}, which is defined as the time elapsed since the most recent update occurred.
To represent the dissatisfaction of users receiving stale data, we introduce the \emph{staleness cost}, which is a non-decreasing function of the AoI (see formal definition in Section~\ref{sec:model}). 
Clearly, one needs to account for both the update cost and the staleness cost when designing an online update policy.

We summarize our main contributions as follows.

\emph{First}, we study the tradeoff between the data freshness and the update cost by formulating an optimization problem to minimize the sum of the staleness cost (which is a function of the AoI) and the update cost.

\emph{Second}, we provide two useful guidelines for the design of optimal update policies.
These guidelines suggest that
$1$) the service provider should update the data only at a point when it receives a request, and $2$) the server should perform an update when the staleness cost is no smaller than the update cost.

\emph{Third}, following these guidelines and assuming that the request arrival process is Bernoulli, we reformulate our problem as a Markov decision process (MDP) and show that there exists a threshold-based policy that is optimal among all online policies, which motivates us to focus on the class of threshold-based policies. Furthermore, we derive the closed-form  expression  of  the average cost under any threshold-based policy and obtain the optimal threshold.

\emph{Finally}, we perform extensive simulations using both synthetic data and real traces to verify our theoretical results and evaluate the performance of our proposed policy compared with several baseline policies. Our simulation results show that the threshold-based policy outperforms the baselines in more general settings (e.g., when the request arrival process is non-Bernoulli).

The remainder of this paper is organized as follows. 
We first discuss related work in  Section~\ref{sec:relatedwork}. 
The system model is described in Section~\ref{sec:model}.
Two guidelines for designing update policies are provided in Section~\ref{sec:algDesign}.
Then, we prove that our MDP formulation admits an optimal threshold-based policy and derive the optimal threshold in Sections~\ref{sec:monotone_optimal_policies}~and~\ref{sec:threshold_optimality}, respectively.
Finally, we present the numerical results in Section~\ref{sec:simulation} and conclude our paper in Section~\ref{sec:conclusion}.

\section{Related Work}\label{sec:relatedwork}
Ever since the concept of AoI was introduced in \cite{kaul12infocom}, the study on the AoI has attracted a  lot of research interest.
There is a large body of work that provides detailed analyses on the AoI performance of information-update systems under different queueing models (M/M/1, M/D/1, etc.) and scheduling policies (FCFS, LCFS, etc.)  \cite{costa2014age, najm2016age}.

Another important line of research focuses on AoI minimization.
One specific type of optimization problem, which is similar to our work, is the joint minimization of AoI and certain costs \cite{8619768,tripathi2021online,9488746}. 
In \cite{8619768}, the authors consider a discrete-time system where an information source is monitored over a communication channel with a transmission cost. They investigate the optimal policy for minimizing the sum of transmission cost and the inaccuracy of the state information at the monitor. It turns out that the optimal policies also have a threshold-based structure. 
Note that in their model, they assume that the source is governed by a random walk process and there are no users, which is different from ours.
A similar source monitoring problem is considered in  \cite{tripathi2021online}, where the goal is to minimize the sum of transmission costs and an unknown, time-varying penalty function of the AoI. They consider both single-source and multi-source scenarios and propose online learning algorithms with provable regret.
In \cite{9488746}, the authors consider a source-monitor pair with stochastic arrival of updates at the source. The source pays a transmission cost to send the update, and its goal is to minimize the weighted sum of AoI and transmission costs. Under the assumption that the update arrival process is Poisson, they propose an optimal threshold-based policy.
Their work differs from ours in their continuous-time setting and no user involvement.

Along this line, researchers have also considered AoI minimization with constraints (see a survey in \cite{9380899}). 
The considered constraints can be viewed as a special type of update cost.
For example, in wireless networks, the update of data consumes wireless channel resources.
Therefore, the number of packets that can be transmitted depends on the interference model \cite{lu2018age}.
Similarly, for caching services, the cache server can only update certain contents at a time due to the capacity constraint  \cite{yates2017age, zhong2018two,bastopcu2021cache,tang2021cache}.
Another example is the energy constraint \cite{8422086,wu2018optimal, 8437573, 9080062}, which is common in  energy-constrained IoT systems.
In these models, the update cost is usually imposed as a constraint of the optimization problem.

While the tradeoff between AoI and costs has been studied, most of them fall into the category where the costs primarily come from service (e.g., CPU cycle and storage) and/or communication (e.g., channel resources, delay, and energy consumption).
In \cite{fountoulakis2020optimal}, the authors consider a system where multiple devices can sample and transmit (or retransmit) updates to one receiver via unstable wireless channels, with each sampling and transmission coming with a sampling cost and a transmission cost, respectively. The objective is to minimize the sum of the expected total sampling costs and transmission costs under the expected AoI constraints. 
A similar problem is also considered in \cite{zhou2019joint}, where the objective is to minimize the expected AoI under the expected energy cost constraint, which is the sum of sampling costs and transmission costs. 
Slightly different from \cite{zhou2019joint}, the work of \cite{huang2020age} studies the problem of minimizing the sum of expected AoI and the expected energy cost under the expected transmission cost constraint. 
We note that those energy costs in \cite{fountoulakis2020optimal,zhou2019joint,huang2020age} are similar to the update cost in our work (especially regarding their mathematical formulation), 
but the origins of costs are slightly different.  
In the applications that we consider, the responses to users’ requests (e.g., gas prices) are usually small and have negligible processing time, but retrieving the data from the information source often requires certain resources (e.g., monetary payment) and introduces update costs. 
More importantly, our work differs from those works in that we emphasize users' perspective and focus on user-perceived data freshness.  In our work, users can proactively query the server to obtain the data (i.e., the ``Pull'' model \cite{sangglobecom17,li2020waiting}), and our primal concern is to optimize the data freshness perceived by the users (which is a penalty function of AoI) rather than at the server. 
The considered pull model and the concern of user-perceived data freshness bring new challenges to the server: upon users' requests, how to balance the tradeoff between the freshness of the data perceived by users and the update costs?

\section{System Model and Problem Formulation}
\label{sec:model}  
 
We consider a discrete-time information-update system that consists of an information source, a service provider (or server for short), and multiple users (see Fig.~\ref{fig:model}). 
The server can communicate with the information source and update its data with the latest information. The users need to query the server to obtain the data.

We consider an aggregated arrival process\footnote{This is because in the applications we consider, the requested information by each user is the same (e.g., in the Gasbuddy application, users who live close by are often interested in the gas prices in the same area), so we can aggregate their requests together.} formed by the
requests from all the users (which will be assumed as Bernoulli
process in Section~\ref{sec:monotone_optimal_policies} for further analysis). 
The requests arrive at the beginning of the time-slot, and the server replies to the requests with the most recently updated data at the end of a time-slot.
We use the metric \textit{Age-of-Information (AoI)} to measure the freshness of data, which is defined as the time elapsed since the most recent update. 
For ease of exhibition, we assume that the AoI drops to $0$ after the update at the end of a time-slot\footnote{Some work also assumes that the AoI drops to $1$ \cite{tripathi2019age,tripathi2021online}. We assume that the AoI drops to $0$ to make the  discussion concise and clear.}. 
The evolution of ${\Delta}(t)$ is as follows:
\begin{equation}
{\Delta}(t)=\left\{\begin{array}{ll}
{\Delta}(t-1)+1, & \text { if } {u}(t)=0; \\
0, & \text { if } {u}(t)=1,
\end{array}\right.
\end{equation}
where ${u}(t)$ indicates whether the server updates the data at time-slot $t$. 
We assume that the time-slot is indexed from $1$ and the initial AoI also equals $1$, i.e., ${\Delta} (1) = 1$.
Let $u_i$ denote the $i$-th update time. Then, an update policy $\pi$ can be denoted by the update times:  $\pi  \triangleq \{ u_i^\pi \} _{i = 1}^\infty $. 
An illustration of a typical AoI evolution is shown in Fig.~\ref{fig:AoI_evolution}. 
To reflect the dissatisfaction level of the users when they receive stale data, we also introduce a \textit{staleness cost} for each response of the server.
Specifically, the staleness cost is defined as a penalty function $f(\Delta)$ of the AoI $\Delta$, where the function $f:[0,\infty ) \mapsto [0,\infty )$ is assumed to be measurable, non-negative, and non-decreasing. 
For simplicity, we let $f(0) = 0$.

At the beginning of each time-slot, 
the server can decide whether to update the data or not. If it does, it needs to pay a constant update cost $p$ and receives the latest data from the information source at the end of time-slot. 
To avoid the staleness cost, the server can first update the data  and then reply to the request with the latest data.
Let $r_j$ be the arrival time of the $j$-th request. 
After the server receives the request at $r_j$, if the server chooses to update the data before replying to the request, its AoI drops to $0$ after the update
and its staleness cost becomes $f(0) = 0$.
In such a case, the server needs to pay an update cost $p$ though.
Otherwise, if the server does not update and replies with the current local data, the server needs to pay a staleness cost $f(\Delta ({r_j}))$. 

\begin{figure}[t]
\centering
\includegraphics[width=0.35\textwidth]{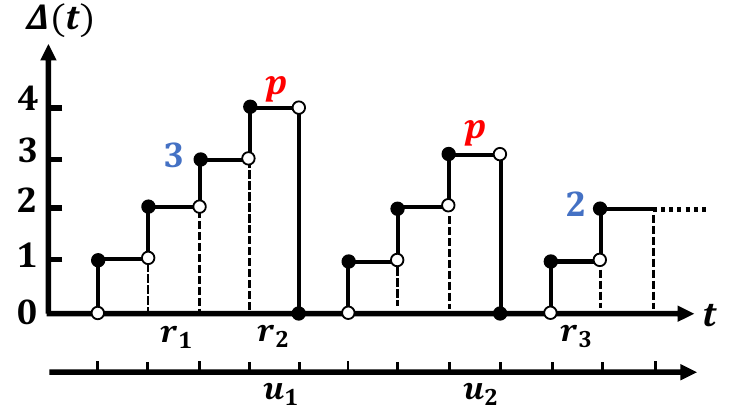}
\caption{An illustration of the AoI evolution at the server.
There are two updates (in time-slots $u_1$ and $u_2$) during the process of serving three requests (in time-slots $r_1$, $r_2$, and $r_3$). 
}
\vspace{-0.3cm}
\label{fig:AoI_evolution}
\end{figure}

Assume that the number of updates during the process of serving $N$ requests under policy $\pi$ is $U^\pi(N)$, i.e.,
\vspace{-0.05cm}
\begin{equation} 
    {U^\pi }(N) \triangleq \max \{ i|{u_i^\pi } \le {r_N}\}.
\vspace{-0.05cm}
\end{equation}
Then, the total cost of serving $N$ requests, which is the sum of update costs and the staleness costs, is defined as
\vspace{-0.05cm}
\begin{equation}
    {C^\pi }(N) \triangleq \sum\limits_{j = 1}^N {f(\Delta ({r_j}))}  + p{U^\pi }(N).
    \vspace{-0.05cm}
\end{equation}
The objective is to find an update policy $\pi$ that minimizes \textit{the long-term average expected cost per request (or average cost for short)}, which is defined as
\begin{equation}
\label{eq:objective}
    {\bar{C}^\pi } \triangleq \mathop {\lim }\limits_{N \to \infty } \frac{\mathbb{E}[{C^\pi }(N)]} {N}, 
\end{equation}
where the expectation is taken over the randomness in the arrival process and the update policy. 
Here, we assume that the limit of average cost under policy $\pi$ exists.
We focus on the set of online policies, denoted by $\Pi$, under which the information available at time $t$ for making update decisions includes the update history, the arrival times of requests that arrive until $t$, and the update cost $p$. Then, we can formulate the following optimization problem:
\begin{equation}
\label{eq:opti_problem}
    \mathop {\min }\limits_{\pi  \in \Pi } {\bar{C}^\pi }. 
\end{equation}

\section{Guidelines for Algorithm Design} 
\label{sec:algDesign}
In this section, we provide two useful guidelines for the design of efficient update policies.
Through a sample-path-dominance argument, 
we show that policies following these guidelines can achieve a lower total cost than those that do not. Therefore, we can reduce the search space of problem~\eqref{eq:opti_problem} to a certain class of online policies.

\subsection{Reactive Policies}
\label{sec:algDesign_reactive}
In this subsection, we present our first guideline for the design of update policies.
As described in Section~\ref{sec:model}, the server can update the data at any time.
However, we show that to achieve a lower total cost, it is sufficient for the update policy to just consider updating the data immediately upon receiving a new request.
We call such policies \textit{Reactive Policies} as the server does not need to update the data when there is no request.
We use $\Pi^R$ to denote the set of reactive policies:
\begin{equation}
 \Pi^R \triangleq \{\pi  \in \Pi \;|\;u_k^\pi  \in \{ {r_j}\} _{j = 1}^\infty ~\text{for all}~ k\}.
\end{equation}
Next, we show that restricting to reactive policies does not incur any performance loss.

\begin{lemma}
\label{lemma:reactive_policy_sufficient}
For any policy $\pi \in \Pi$, there exists a reactive policy $\pi^{\prime} \in \Pi^{R}$ that achieves an average cost no larger than that of policy $\pi$, i.e., $\bar{C}^{\pi^{\prime}} \leq \bar{C}^{\pi}$. 
\end{lemma}

We provide the detailed proof in Appendix~\ref{appendix:reactive_policy_proof} and explain the key ideas as follows.
Intuitively, postponing the update until a request arrives does not increase the total cost because the total number of updates remains the same, but doing this achieves a lower staleness cost since the update time is closer to the request arrival time. 
Therefore, reactive policies can achieve a smaller total cost than those non-reactive policies.

Lemma~\ref{lemma:reactive_policy_sufficient} implies that the search space of Problem \eqref{eq:opti_problem}  can be further reduced from the set of online policies $\Pi$ to the set of reactive policies $\Pi^R$.
Now, consider any reactive policy $\pi$. Upon receiving a request, the server needs to decide whether to update the data or not before responding.
Therefore, we use $I_j^\pi $ to denote the decision made by the server upon receiving the $j$-th request for the data at time $r_{j}$:
\begin{equation}\notag
I_{j}^\pi \triangleq \left\{
\begin{aligned}
1, & \quad \text{if the server updates the data at time}~ r_{j};\\
0, & \quad \text{otherwise}.
\end{aligned}
\right.
\end{equation}

\subsection{Capped Reactive Policies}
In this subsection, we present the second guideline for the design of update policies.
In Section \ref{sec:algDesign_reactive}, we show that the reactive policies achieve a smaller or equal average cost by postponing the update until a request arrives. 
In fact, after the server receives the request, if the staleness cost is no smaller than the update cost, it is better for the server to update the data to avoid a larger staleness cost.
Doing so not only leads to a smaller cost for this request but also benefits the next few requests. 
We use $\Pi^{R+}$ to denote the set of reactive policies that satisfy the above guideline:
\begin{align}
\Pi^{R+} \triangleq & ~\{\pi\in\Pi^{R} ~|~  {I}_j^\pi = 1 ~\text{for all}~ j ~\text{when}~ f(\Delta(r_j)) \geq p\}.\notag
\end{align}
That is, for any policy $\pi  \in {\Pi ^{R + }}$, it must update the data when the staleness cost is no smaller than the update cost; otherwise, it can choose to update the data or not.
We call such policies \textit{Capped Reactive Policies} because the staleness cost of such policies is capped by the update cost.
Fig.~\ref{fig:policies_relation} illustrates the relationship between $\Pi^{R}$ and $\Pi^{R+}$.
Note that the condition ${f}({\Delta}({r_{j}})) \ge {p}$ can also be expressed as ${\Delta}({r_{j}}) \ge \Delta^*$, where  $\Delta^*$ is
the smallest AoI such that the staleness cost is no smaller than the update cost, i.e., 
\begin{equation}
\label{eq:smallest_s}
    {\Delta^*} \triangleq \min \{ \Delta|f(\Delta) \ge p\}.
\end{equation}
In the following, we show that restricting to capped reactive policies does not incur any performance loss.

\begin{lemma}
\label{lemma:r+}
For any policy $\pi \in \Pi^{R}$, there exists a capped reactive policy $\pi^{\prime} \in \Pi^{R+}$ that achieves an average cost no larger than that of policy $\pi$, i.e., $\bar{C}^{\pi^{\prime}}\leq \bar{C}^{\pi}$. 
\end{lemma}

We provide the detailed proof in Appendix~\ref{appendix:capped_reactive_policy_proof} and explain the key ideas in the following.
Upon receiving a request, policy $\pi^{\prime}$ performs an update if the staleness cost is no smaller than the update cost. Compared to policy $\pi$ that does not make such an update, doing so incurs an update cost for policy $\pi^{\prime}$, but it avoids a larger staleness cost. Besides, it also reduces the staleness cost for the requests that arrive thereafter. Therefore, policy $\pi^{\prime}$ can achieve a total cost no larger than that of policy~$\pi$.

By Lemma~\ref{lemma:r+}, we can further reduce the search space to the class of capped reactive policies $\Pi^{R+}$.
Therefore, Problem~\eqref{eq:opti_problem} can be further reduced to the following: 
\begin{equation}\label{eq:optformula3}
\min_{\pi \in \Pi^{R+}}~ \bar{C}^{\pi}.
\end{equation}

\begin{figure}
    \centering
    \includegraphics[scale=0.6]{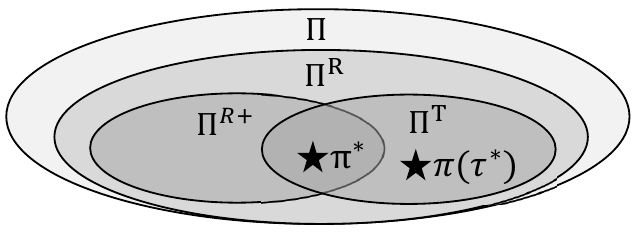}
    \caption{The relationship between online policies $\Pi$, reactive policies $\Pi^{R}$, capped reactive policies $\Pi^{R+}$, threshold-based policies $\Pi^{T}$, an overall optimal policy $\pi^*$ (see Theorem~\ref{thm:exi_opt_avg_mdp}), and an optimal threshold-based policy ${\pi ({\tau ^*})}$ (see Corollary~\ref{lemma:opt_threshold}). Note that policy ${\pi ({\tau ^*})}$ is also an overall optimal policy and could be the same as policy $\pi^*$ in some cases.}
    \label{fig:policies_relation}
    \vspace{-0.3cm}
\end{figure}

Till this point, we do not make any assumption on the request arrival process.
The aforementioned guidelines can be applied to general request arrival processes.
In the following, unless otherwise specified, we focus on the capped reactive policies $\Pi^{R+}$. 
\emph{This capped property plays an important role in characterizing the threshold-based structure of an optimal policy for solving the MDP formulation in Section~\ref{sec:monotone_optimal_policies}.}

\section{MDP Formulation and Threshold structure }
\label{sec:monotone_optimal_policies}

Under a capped reactive policy, the server makes update decisions upon receiving requests and pays a cost (an update cost or a staleness cost) based on the decision. 
Naturally, this sequential decision process can be modeled as an MDP.
In this section, 
we assume that the request arrival process is Bernoulli\footnote{The motivation for this assumption is that the probabilistic analysis shows that the arrivals in certain arrival-type processes (i.e., users' arrival to a bank in each second, and job arrivals to the server at each time slot) are independent random variables, which can be modeled as independent Bernoulli trials \cite{bertsekas2008introduction, gallager2013stochastic}. 
In addition, the Bernoulli arrival process is actually a discrete-time analog of the Poisson arrival process, which is widely used in modeling the arrival process in real life \cite{bertsekas2008introduction, gallager2013stochastic, harchol2013performance}. 
Therefore, it also makes sense to model the arrival process in the discrete-time system as a Bernoulli process.
This assumption is also widely adopted in recent AoI-related work \cite{akar2021discrete,tripathi2017age,tripathi2019age,kosta2021age}.} with rate $\lambda \in (0,1)$, denoted by Bernoulli($\lambda$), i.e., the probability that a request arrives in a time-slot is $\lambda$. 
Then, we reformulate Problem~\eqref{eq:optformula3} as a discrete-time MDP and show that there exists a stationary threshold-based policy that solves the Bellman equation of the considered MDP and is thus optimal among all online policies.

The MDP formulation has the following key components:
$\{\mathcal{N}, \mathcal{S}, \mathcal{A}_s, p(\cdot \mid s, a), c(s, a): n \in \mathcal{N}, s \in \mathcal{S}, a \in \mathcal{A}_s\}$, where
\begin{enumerate}
	\item  
	$\mathcal{N} = \{ 1,2, \cdots \} $
	is the set of decision epochs. 
	Under a capped reactive policy, the $n$-th decision epoch is the time-slot when the $n$-th request arrives.
	\item $\mathcal{S}=\{ 0,1,\cdots\}$ 
	is the set of system states (which are all possible values of the AoI). We use $s_n$ to denote the AoI value when the $n$-th request arrives.  
	\item 
	$\mathcal{A}_s$ is the set of actions when the system state is $s$. 
	Let $a \in {\mathcal{A}_s}$ denote the possible actions, where $a = 1$ means updating the data and $a = 0$ means not. 
	Under a capped reactive policy, there are two sets of actions depending on the state $s$: when the staleness cost $f(s)$ is no smaller than the update cost $p$, the only available action is to update, i.e., ${\mathcal{A}_{\{ s:f(s) \ge p\} }} = \{ 1\} $;
	otherwise, the system can either update or not, i.e., ${\mathcal{A}_{\{ s:f(s) < p\} }} = \{ 0,1\} $. 
	\item The transition probability can be calculated as
\begin{equation*}
\begin{aligned}
&p(z \mid s, a)= \\
&\left\{\begin{array}{ll}
(1-\lambda)^{z-1} \lambda, & \ \  \text {if} \ z \geq 1 \text { and } a=1; \\
(1-\lambda)^{z-s-1} \lambda, & \ \ \text {if} \ z>s, f(s)<p, \text { and } a=0; \\
0, & \  \text { otherwise. }
\end{array}\right.
\end{aligned}
\end{equation*}
	That is, when the system is in state $s$, if the server updates the data, the system will enter state $z$ $(z\geq1)$ with probability ${{{(1 - \lambda )}^{z - 1}}\lambda }$ because the request arrival process follows Bernoulli($\lambda$); otherwise, if the server does not update, under a capped reactive policy, the system will enter state $z$ $(z>s)$ with  probability ${{{(1 - \lambda )}^{z - s - 1}}\lambda }$ only when the staleness cost $f(s)$ is smaller than the update cost $p$. 
	\item The cost at each decision epoch can be expressed as 
	\begin{equation}
		c(s,a) = \left\{ {\begin{array}{*{20}{ll}}
p, &  \text{if} \ a = 1;\\
f(s), &  \text{if} \  a = 0.
\end{array}} \right.
	\end{equation}
	That is, when the system is in state $s$, updating the data incurs an update cost of $p$; otherwise, there is a staleness cost of $f(s)$.
	Note that under a capped reactive policy, we always have $c(s,a) \le p$.
\end{enumerate}

The objective of the MDP is to find a stationary capped reactive stationary update policy that minimizes the long-term average expected cost, i.e., 
\begin{equation}
    \label{eq:obj_ini}
    \mathop {\min }\limits_{\pi  \in {\Pi ^{R + }}} \mathop {\lim }\limits_{N \to \infty } \frac {\mathbb{E}_\pi\left[ {\sum\limits_{n = 1}^N {c({s_n},{a_n})} }| s_1 = s \right]} {N},
\end{equation}
where $\mathbb{E}_\pi[ \cdot ]$ represents the conditional expectation, given that policy $\pi$ is employed; ${s_n}$ and ${a_n}$ are the state  and action taken at decision epoch~$n$, respectively; and $s$ is the initial state. 
We emphasize that, unlike traditional MDP formulations that mainly focus on optimization over time (where the state of the MDP should consist of two variables: one denotes the AoI value in the slot and the other denotes whether there is a request arriving in the slot or not. Obtaining the simple solutions (e.g., threshold structure) for the MDPs with multiple state variables is usually more challenging and involves more sophisticated techniques.),  by following our proposed guidelines (especially the reactive guideline), we can optimize our MDP over only users’ requests (so now our state includes just one variable: the AoI value when the request arrives). 
This allows us to reduce the state space and thus facilitate the theoretical analysis.
Note that the objective in Problem~\eqref{eq:obj_ini} is the same as that in Problem~\eqref{eq:optformula3}  except that we specify the initial state $s$ in Problem~\eqref{eq:obj_ini}. 
In other words, an optimal policy for Problem~\eqref{eq:obj_ini} is also an optimal policy for Problem~\eqref{eq:optformula3}.
Next, we show that there exists an optimal policy for Problem~\eqref{eq:obj_ini} that has a threshold-based structure, which enables us to search for an optimal policy in the class of threshold-based policies (see Section~\ref{sec:threshold_optimality}). We state this result in Theorem~\ref{thm:exi_opt_avg_mdp}.

\begin{theorem}
\label{thm:exi_opt_avg_mdp}
There exists an optimal stationary capped reactive policy $\pi^* \in {\Pi ^{R + }}$ that has a threshold-based structure.
\end{theorem}

We provide the detailed proof in Appendix~\ref{appendix:theorem_1} and present an outline of the proof in the following.
First, we study a discounted MDP and derive its optimal value function.
Second, based on the optimal value function, we derive the Bellman equation of the expected total average cost and show that the Bellman equation has a threshold-based structure. Specifically, the server needs to update the data when the current AoI value (i.e., the state) is no smaller than a certain fixed threshold $s^*$ (see definition in Eq.~\eqref{eq:min_i}); otherwise, it does not.
Now consider a stationary capped reactive policy $\pi^* \in {\Pi ^{R + }}$ that makes update decisions based on threshold $s^*$. Apparently, policy $\pi^*$ minimizes the Bellman equation for any state, thus it is an overall optimal policy~\cite[Chapter V, Theorem $2.1$]{ross2014introduction}. 
The threshold structure of the optimal policy $\pi^*$ indicates that among all threshold-based policies, there is an overall optimal policy (see Fig.~\ref{fig:policies_relation}). This motivates us to search for the optimal threshold-based policy in Section~\ref{sec:threshold_optimality}. 

\textit{Remark $1$:} Our proposed guidelines (especially the capped reactive policy) play an important role in characterizing the threshold-based structure of an optimal policy. By following our guidelines, we can restrict ourselves to the policies whose cost at each decision epoch is no greater than the update cost $p$. 
This additional property enables us to characterize the monotonicity of the optimal value function of the discounted MDP and the monotonicity of the Bellman equation, and ultimately address the overall problem by finding a simple threshold-based optimal policy.

\section{Optimal Threshold-based Policy}
\label{sec:threshold_optimality}

Theorem~\ref{thm:exi_opt_avg_mdp} tells us that we can further reduce the search space from the set of capped reactive policies to the set of capped reactive threshold-based policies.
In this section, we formally define threshold-based policies and derive the closed-form expression of the average cost of the threshold-based policies. Using the closed-form expression, we can find the optimal threshold-based policy. 
Furthermore, we show that the optimal threshold-based policy is also an optimal policy among all online policies.

We begin with the definition of threshold-based policies. 
\begin{definition}[Threshold-based Policies]
A policy in $\Pi^{R}$ is called a threshold-based policy if it performs updates according to the following rule with a predetermined positive integer threshold~$\tau$: for the request arriving at time ${r_j}$, we have
\begin{equation}\notag
I_j = \left\{
\begin{aligned}
1, &\quad \Delta(r_j) \ge \tau; \\
0, &\quad \Delta(r_j) < \tau. \\
\end{aligned}
\right.
\end{equation} 
That is, the server updates the data at $r_j$ before replying if the AoI at $r_j$ is no smaller than threshold $\tau$; 
otherwise, the server simply replies with the current local data.
\end{definition}

We consider an integer threshold because the values of the AoI are integers. 
Let $\pi (\tau )$ be the threshold-based policy with threshold $\tau$, and let ${\Pi ^T}$ be the set of all threshold-based policies.
Fig.~\ref{fig:policies_relation} shows the relationship of $\Pi^{R}$, $\Pi^{R+}$, and $\Pi ^T$.

Assume that the request arrival process is Bernoulli, we can derive the closed-form expression of the average cost under any threshold-based policy. We state this result in Theorem~\ref{theorem:thresholdperformance}.

\begin{theorem}
\label{theorem:thresholdperformance}
Assume that the request arrival process is Bernoulli($\lambda$), the staleness cost function is $f(\Delta )$, and the update cost is $p$. Then, for any policy $\pi(\tau)  \in {\Pi ^T}$ with a positive integer threshold $\tau$, the average expected cost can be computed as follows:
\begin{equation}
\label{eq:exp_avg_cost}
    {\bar{C}^{\pi(\tau)} } = \frac{{\lambda \sum\limits_{t = 1}^{\tau  - 1} {f(t)}  + p}}{{\lambda (\tau  - 1) + 1}}.
\end{equation}
\end{theorem}
We provide the detailed proof in Appendix~\ref{appendix:theorem_2} and present an outline of the proof in the following. 
Since the request arrival process is Bernoulli, under a threshold-based update policy, the lengths of update intervals are independent and identically distributed  (\textit{i.i.d.}). Thus, the update process is a renewal process. 
Due to the ergodicity of the process,  the expected average cost can be computed as ${\bar{C}^{\pi(\tau)} } = \mathbb{E}[{C_k}]/\mathbb{E}[{N_k}]$, where $\mathbb{E}[{C_k}]$ and $\mathbb{E}[{N_k}]$ are the expected total cost and the expected number of requests in the $k$-th update interval, respectively. 
 In addition, by exploiting the properties of Bernoulli arrival process, we can further derive the closed-form expressions of $\mathbb{E}[{C_k}]$ and $\mathbb{E}[{N_k}]$, which are shown in the numerator and denominator of Eq.~\eqref{eq:exp_avg_cost}, respectively.

With the result in Theorem~\ref{theorem:thresholdperformance}, we can also easily compute the optimal threshold ${\tau ^*}$. 
In fact, this optimal threshold-based policy ${\pi(\tau^*)}$ is also an overall optimal policy, which is shown in Corollary~\ref{lemma:opt_threshold}.

\begin{corollary}
\label{lemma:opt_threshold}
Assume that the request arrival process is Bernoulli($\lambda$). Then, the threshold of the optimal threshold-based policy ${\pi(\tau^*)} \in  {\Pi ^T}$ is the following:
\begin{equation}
    {\tau ^*} =  \left\{ {\mathop {\arg \min }\limits_{\tau  \in \{ \left\lfloor {\tau '} \right\rfloor ,\left\lceil {\tau '} \right\rceil \} } {\bar{C}^{\pi (\tau )}} } \right\},
    \label{eq:opt_threshold}
\end{equation}
where $\tau'$ is the real number that achieves the smallest expected average cost (i.e., $\tau ' = \arg \min_{\tau>0} {\bar{C}^{\pi(\tau)} }$). 
Furthermore, the optimal threshold-based policy ${\pi(\tau^*)} \in  {\Pi ^T}$ is also an overall optimal policy among all online policies.
\end{corollary}

\begin{proof}
Theorem~\ref{thm:exi_opt_avg_mdp} states that there exists an overall optimal capped reactive threshold-based policy ${\pi ^*} \in {\Pi ^{R + }} \cap {\Pi ^T}$.
For the optimal threshold-based policy ${\pi(\tau^*)} \in  {\Pi ^T}$, we have ${\bar{C}^{\pi ({\tau ^*})}} \le {\bar{C}^{{\pi ^*}}}$. 
On the other hand, since policy ${\pi ^*}$ is an overall optimal policy, we also have ${\bar{C}^{\pi ({\tau ^*})}} \ge {\bar{C}^{{\pi ^*}}}$. Therefore, we have ${\bar{C}^{\pi ({\tau ^*})}} = {\bar{C}^{{\pi ^*}}}$. This implies that the optimal threshold-based policy ${\pi(\tau^*)} \in  {\Pi ^T}$ is also an overall optimal policy among all online policies.
\end{proof}

Here, the real number $\tau'$ in Eq.~\eqref{eq:opt_threshold} can either be theoretically calculated if the expression of $\sum\nolimits_{t = 1}^{\tau  - 1} {f(t)} $ in ${\bar{C}^{\pi(\tau)} }$ is known (see below for two examples) or be numerically calculated otherwise.
The optimal threshold ${\tau ^*}$ may not be unique, depending on the staleness cost function $f$. Also, policy ${\pi ({\tau ^*})}$ could be the same as policy $\pi^*$ in some cases.

In the following, we provide the average expected cost and optimal threshold when the staleness cost is a linear function and a square function of the AoI, respectively. 
\begin{example}[A linear staleness cost function]
\label{example:linear}
Assume that the request arrival process is Bernoulli($\lambda$), $f(\Delta ) = \Delta $, and the update cost is $p$. Then, for any policy $\pi (\tau) \in {\Pi ^T}$ with a positive integer threshold $\tau$,  we have
\begin{equation}
    {\bar{C}^{\pi(\tau)} } = \frac{{\lambda \tau (\tau  - 1)/2 + p}}{{\lambda (\tau  - 1) + 1}},
    \label{eq:exampel_1}
\end{equation}
and the optimal threshold ${\tau ^*} = \left\{ {\arg {{\min }_{\tau  \in \{ \left\lfloor {\tau '} \right\rfloor ,\left\lceil {\tau '} \right\rceil \} }}{\bar{C}^{\pi (\tau )}}} \right\}$,
where  $\tau ' = (\sqrt {2p\lambda  - \lambda  + 1}  + \lambda  - 1)/\lambda $.
\end{example}

\begin{example}[A quadratic staleness cost function]
Assume that the request arrival process is Bernoulli($\lambda $), $f(\Delta ) = \Delta^2$, and the update cost is $p$. Then, for any policy $\pi (\tau) \in {\Pi ^T}$ with a positive integer threshold $\tau$,  we have 
\begin{equation}
    {\bar{C}^{\pi(\tau)} } = \frac{{\lambda \left[ {{{(\tau  - 1)}^3}/3 + {{(\tau  - 1)}^2}/2 + (\tau  - 1)/6} \right] + p}}{{\lambda (\tau  - 1) + 1}},
    \label{eq:exampel_2}
\end{equation}
and the optimal threshold ${\tau ^*} = \left\{ {\arg {{\min }_{\tau  \in \{ \left\lfloor {\tau '} \right\rfloor ,\left\lceil {\tau '} \right\rceil \} }}{\bar{C}^{\pi (\tau )}}} \right\}$,
where  $\tau '$ is the solution of 
\begin{equation*}
    1 - 6p - 6\tau  + 6{\tau ^2} + \lambda (4\tau  - 1){(\tau  - 1)^2} = 0.
\end{equation*}
\end{example}

\section{Numerical and Experimental Results}\label{sec:simulation}

In this section, we perform extensive simulations and experiments to verify our theoretical results and compare the performance of the optimal threshold-based policy with several baseline policies using both synthetic data and real traces. 
Throughout this section, we consider two types of the staleness cost: a linear function (i.e., $f(\Delta ) = \Delta $) and a quadratic function (i.e., $f(\Delta ) = \Delta^2$).

We first evaluate the performance of threshold-based policies with different thresholds when the staleness cost is a linear function of AoI in Fig.~\ref{fig:multithresholds}.
The setting of the simulations is as follows.
The request arrival process is Bernoulli with rate $\lambda=0.1$, and the update cost is $p=100$. 
The simulation results are the average of 100 simulation runs, where each run consists of $N= 10^4$ requests (which is our default setting for the synthetic simulations).
We also include a breakdown of the results in terms of average staleness cost $\bar{a} \triangleq \sum_{j = 1}^{N} f(\Delta(r_j))/N$ and average update cost $\bar{p} \triangleq pU^\pi(N)/N$.
We observe that the simulation results of average total cost under threshold-based policies perfectly match the theoretical results in Example~\ref{example:linear}. Clearly, as the threshold increases, the update cost decreases, but the staleness cost increases. This is as expected because a higher threshold leads to less frequent updates, which results in a smaller update cost but a larger staleness cost. 
As a result, the average total cost, which is the sum of the two, first decreases and then increases. 
The optimal cost ${C^{\pi ({\tau ^*})}} \approx 36.22$ is achieved at ${\tau ^*} =  \{\arg {\min _{\{ \left\lfloor {\tau '} \right\rfloor ,\left\lceil {\tau '} \right\rceil \} }}{C^{\pi (\tau )}}\}  = 37$, where $\tau ' = (\sqrt {2p\lambda  - \lambda  + 1}  + \lambda  - 1)/\lambda  \approx 36.72$. 
Similar observations can also be made from Fig.~\ref{fig:square_multithresholds}, where the staleness cost is a quadratic function of AoI.
As expected, the staleness cost increases remarkably with the threshold, since the staleness cost function is quadratic.

Next, we compare the performance of the optimal threshold-based policy with several baselines in Figs.~\ref{fig:algcomp} and \ref{fig:square_algcomp}, where the staleness cost function is linear and quadratic, respectively.
We consider three baselines: (i) a naive policy, (ii) periodic policies, and (iii) the optimal offline policy.
The naive policy is a capped reactive threshold-based policy with a threshold being
equal to ${\Delta ^*} = \left\lceil p \right\rceil$ when $f(\Delta ) = \Delta $ (or ${\Delta ^*} = \left\lceil \sqrt{p} \right\rceil$ when $f(\Delta ) = \Delta ^2$).
That is, upon receiving a request, this policy naively updates the data when the staleness cost is no smaller than the update cost, otherwise it does not.
A periodic policy has a positive integer period $d$ and updates the data every $d$ time-slots, i.e., $u_i = id$ for $i = 1, 2,\dots$. Note that a periodic policy is not a reactive policy.
Following a similar argument in the proof of Theorem~\ref{theorem:thresholdperformance}, we can show that the average cost under a periodic policy with period $d$ is $({p} + \lambda {d}({d} - 1)/2)/\lambda {d}$ when $f(\Delta ) = \Delta $ (or $(p + \lambda ({(d - 1)^3}/3 + {(d - 1)^2}/2 + (d - 1)/6))/\lambda d$ when $f(\Delta ) = \Delta^2 $). 
In the comparisons, we only consider the optimal periodic policy with  $d^* = \left\lceil {\sqrt {2{p}/\lambda} } \right\rceil $ when $f(\Delta ) = \Delta $ (or $d^*$ being the solution of  $4\lambda {d^3} - 3\lambda {d^2} - 6p = 0$ when $f(\Delta ) = \Delta^2 $).
The optimal offline policy has the exact knowledge of all the request arrival times and is obtained based on the dynamic programming approach.
Hence, the average cost under an optimal offline policy can be viewed as a lower bound of all online policies.

\begin{figure}[!t]
		\centering
        \includegraphics[width=0.35\textwidth]{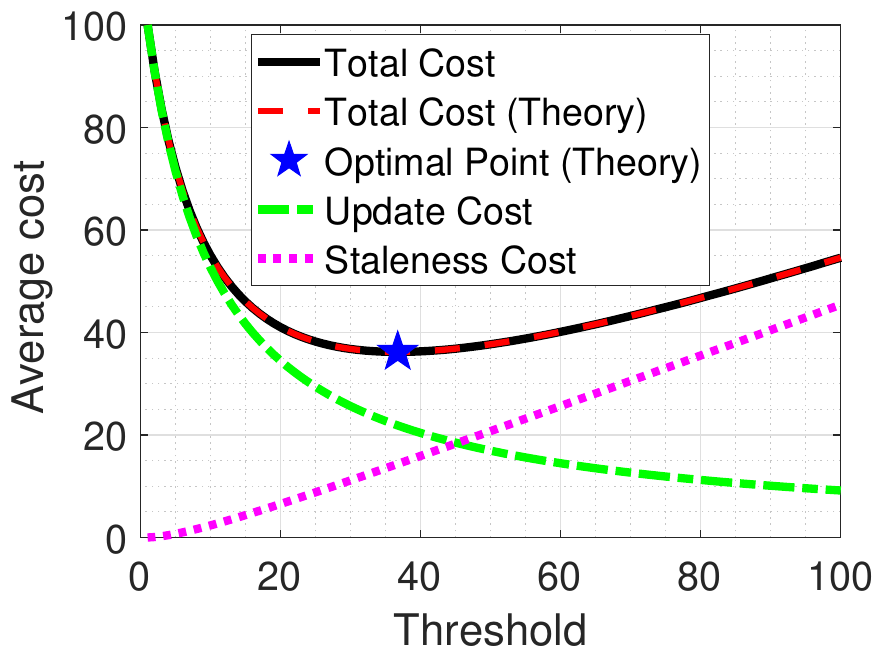}
		\caption{Average cost under threshold-based policies with different thresholds when $f(\Delta ) = \Delta $, where $\lambda  = 0.1$ and $p = 100$.}
	\label{fig:multithresholds}
\vspace{-0.3cm}
\end{figure}

\begin{figure}[!t]
		\centering
        \includegraphics[width=0.35\textwidth]{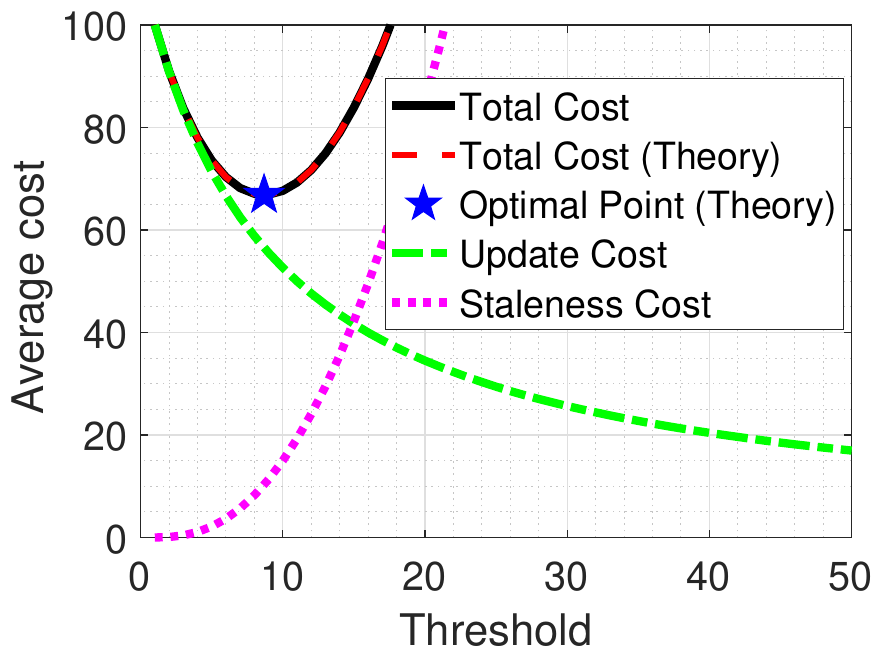}
		\caption{Average cost under threshold-based policies with  different thresholds when $f(\Delta ) = \Delta^2$, where $\lambda  = 0.1$ and $p = 100$.}
	\label{fig:square_multithresholds}
\vspace{-0.3cm}
\end{figure}

\begin{figure*}[!t]
	\begin{minipage}[t]{0.66\textwidth}
		\centering
		\begin{subfigure}[b]{0.49\linewidth}
			\includegraphics[width=1\textwidth]{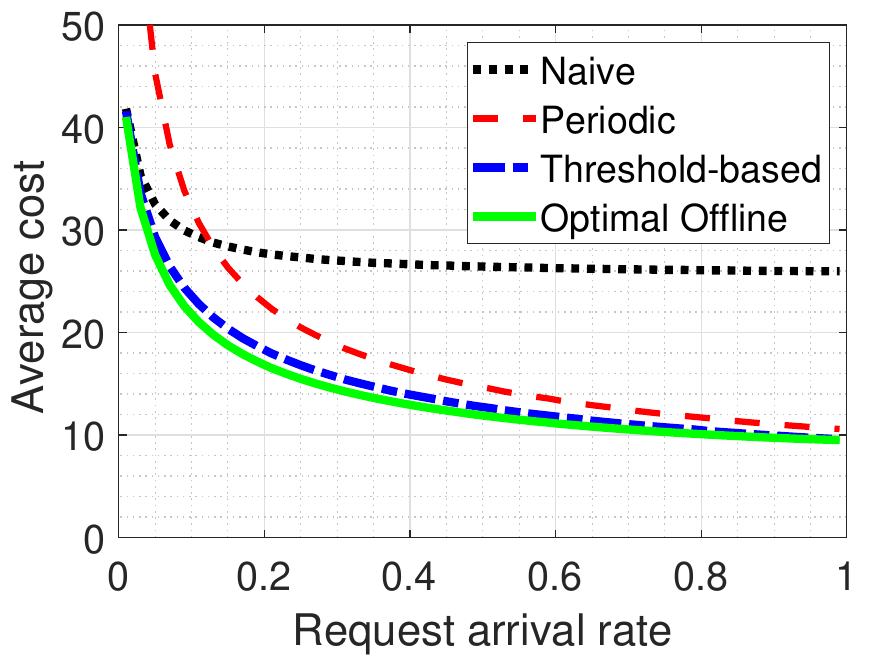}
			\caption{Fixed update cost $p = 50$}
			\label{fig:multirates}
		\end{subfigure}
		\begin{subfigure}[b]{0.49\linewidth}
		\includegraphics[width=1\textwidth]{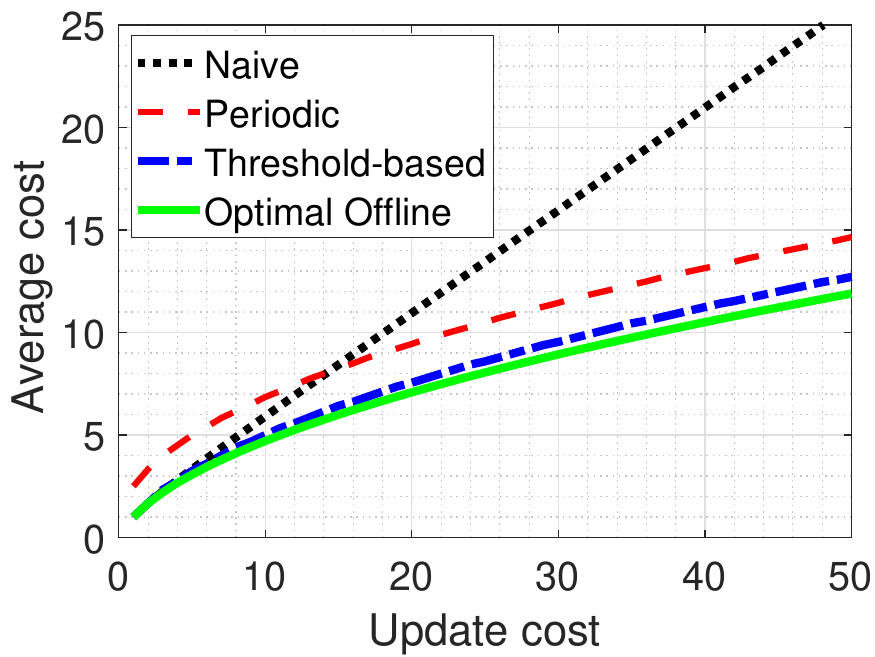}
			\caption{Fixed request  arrival rate $\lambda = 0.5$}
			\label{fig:multicosts}
		\end{subfigure}
		\caption{Performance comparisons of different policies with different request arrival rate $\lambda$ and different update cost $p$ when $f(\Delta ) = \Delta$, respectively.}
		\label{fig:algcomp}
	\end{minipage}
	\quad
	\begin{minipage}[t]{0.33\textwidth}
	\centering
    \begin{subfigure}[b]{1\linewidth}
    \centering
        \includegraphics[width=0.98\textwidth]{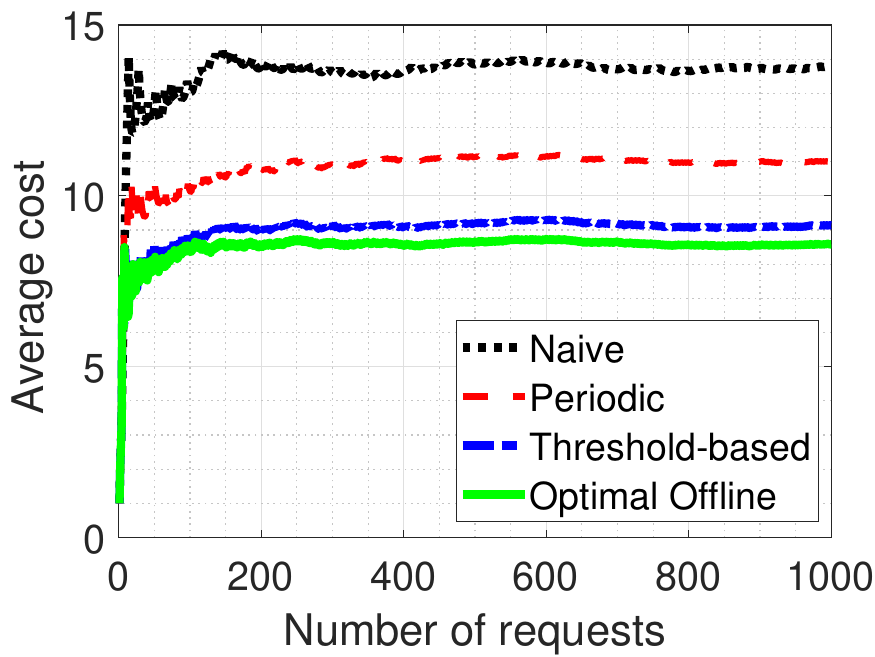}
        \caption*{}
    \end{subfigure}
    \vspace{-1cm}
    \caption{Performance comparisons of different policies using trace dataset,
    where update cost $p=25$, request arrival rate $\lambda = 0.4$, and $f(\Delta ) = \Delta$.
    }
    \label{fig:trace_res}
	\end{minipage}
	\vspace{-0.3cm}
\end{figure*}

\begin{figure*}[!t]
	\begin{minipage}[t]{0.66\textwidth}
		\centering
		\begin{subfigure}[b]{0.49\linewidth}
			\includegraphics[width=1\textwidth]{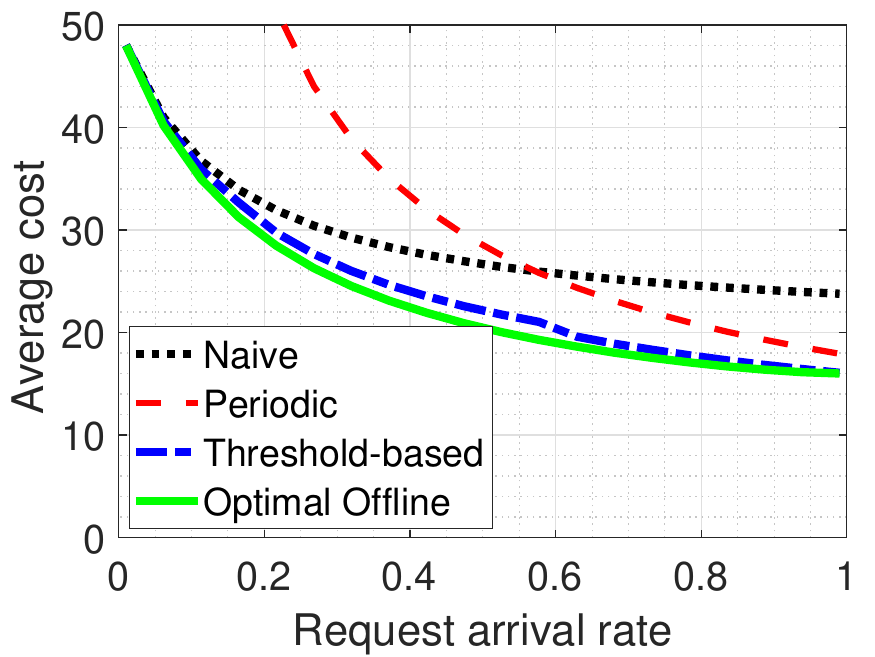}
			\caption{Fixed update cost $p = 50$}
			\label{fig:square_multirates}
		\end{subfigure}
		\begin{subfigure}[b]{0.49\linewidth}
		\includegraphics[width=1\textwidth]{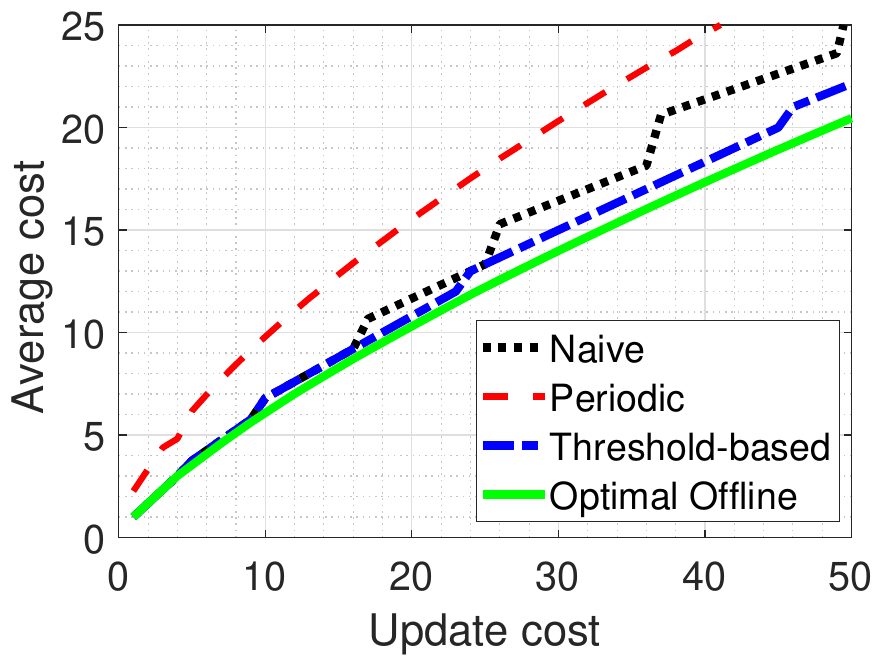}
			\caption{Fixed request  arrival rate $\lambda = 0.5$}
			\label{fig:square_multicosts}
		\end{subfigure}
		\caption{Performance comparisons of different policies with different request arrival rate $\lambda$ and different update cost $p$ when $f(\Delta ) = \Delta^2 $, respectively.}
		\label{fig:square_algcomp}
	\end{minipage}
	\quad
	\begin{minipage}[t]{0.33\textwidth}
	\centering
    \begin{subfigure}[b]{1\linewidth}
    \centering
        \includegraphics[width=0.98\textwidth]{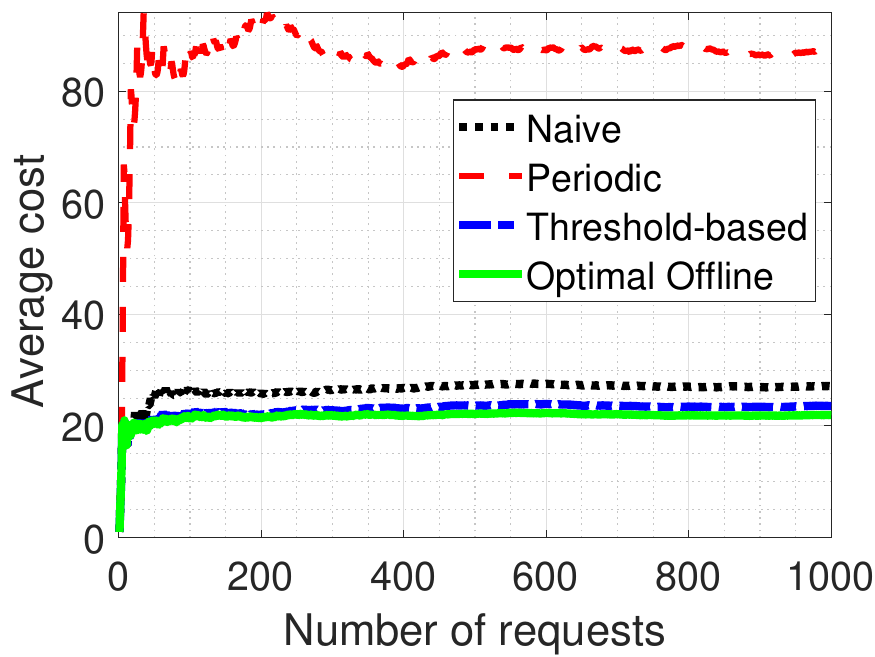}
        \caption*{}
    \end{subfigure}
    \vspace{-1cm}
    \caption{Performance comparisons of different policies using trace dataset,
    where update cost $p=50$, request arrival rate $\lambda = 0.4$, and $f(\Delta ) = \Delta^2$.
    }
    \label{fig:square_trace_res}
	\end{minipage}
	\vspace{-0.3cm}
\end{figure*}

Fig.~\ref{fig:multirates} shows the results for the setting with a fixed update cost $p = 50$ and a varying  request arrival rate $\lambda$ when the staleness cost function is linear. 
We observe that the optimal threshold-based policy outperforms all the other online policies and is very close to the optimal offline policy.
When the request arrival rate is small, the optimal periodic policy performs poorly. This is because Bernoulli process with a small rate $\lambda$  results in a larger mean (i.e., $1/\lambda$) and a larger variance (i.e., $(1 - \lambda)/{\lambda^2}$) of the inter-arrival time of the requests. Hence, the inter-arrival time of the requests is usually larger and more random.
In this case, a periodic policy that updates the data at fixed time instants resulting in a larger staleness cost. 
On the other hand, the interarrival time of requests is large when the rate $\lambda$ approaches $0$, resulting in a large AoI (and thus a high staleness cost) when the requests arrive. All the capped reactive policies as well as the offline optimal policy would make an update decision for each request, making their average cost close to the update cost $p$. This aligns well with our theoretical result when we let $\lambda =0$ in Eq.~\eqref{eq:exampel_1}.
When the request arrival rate becomes larger, the performance of the optimal periodic policy improves and is close to that of the optimal threshold-based policy. This is because a large request arrival rate $\lambda$ leads to a small period $d^*$. In this case, the requests arrive more frequently, and the updates also occur more frequently. Hence, the staleness cost becomes small, and the update cost becomes dominant.
We also observe that the naive policy performs poorly compared to the other policies when $\lambda$ is large.
This is because the naive policy is agnostic about the request arrival rate. The update period under a naive policy is roughly equal to $p$ regardless of the request arrival rate. However, when $\lambda$ becomes large, there could be many more requests arriving during an update interval of length $p$, which results in a large staleness cost.
In addition, when the request arrival rate $\lambda$ approaches $1$, the optimal threshold-based policy and the optimal periodic policy have the same optimal threshold/period (i.e., ${\tau ^*} = {d^*} = \left\lceil {\sqrt {2p} } \right\rceil $) as well as the same average cost (i.e., $\sqrt {2p}  - 1/2$), 
and their knowledge about the future request arrival process is almost the same as the optimal offline policy (i.e., there is a request arriving in every time-slot),  so their performance is also very close. 
Similar observations can also be made from Fig.~\ref{fig:square_multirates}, where the staleness cost function is quadratic. We omit the explanation since they are almost identical.

Fig.~\ref{fig:multicosts} shows the results for the setting with a fixed  request arrival rate $\lambda=0.5$ and a varying update cost $p$ when the staleness cost function is linear. We observe that the optimal threshold-based policy again outperforms all the other online policies and performs closely to the optimal offline policy. 
The optimal periodic policy performs poorly because it  is not a reactive policy, it cannot update timely according to the staleness cost of the requests, resulting in a large total staleness cost. We also observe that as the update cost increases, the naive policy performs much worse compared to the other policies. This is because the average cost under a naive policy increases at a rate of $O(p)$,
while the average cost under the other policies increases at a rate of $O(\sqrt{p})$.
In Fig.~\ref{fig:square_multicosts}, we show the results under the same setting except that the staleness cost function is quadratic. Again, the optimal threshold-based policy outperforms all other online policies. The optimal periodic policy performs poorly at all update costs. This is because its staleness costs become much larger due to the quadratic function and its inability to update timely upon the arrival of requests.
The average cost under the naive policy has several jumps. The reason is that when the threshold ${\Delta ^*}$ is set to $\left\lceil \sqrt{p} \right\rceil$, the integral change of $\left\lceil \sqrt{p} \right\rceil$ leads to a high increase in the average cost. For example, when the update cost $p$ in the range of $[26,36]$, the threshold ${\Delta ^*}$ under the naive policy is $\left\lceil {\sqrt p } \right\rceil  = 6$ and the average cost increases linearly with respect to $p$ according to Eq.~\eqref{eq:exampel_2}. However, when the update cost is in the range of $[37,49]$, the threshold ${\Delta ^*}$ under the naive policy becomes $\left\lceil {\sqrt p } \right\rceil  = 7$, which results in a much larger average cost compared to when ${\Delta ^*} = 6$ according to Eq.~\eqref{eq:exampel_2}. This also explains the jumps in the optimal threshold-based policy, but the jump is almost unnoticeably (e.g., when $p=46$) since it chooses the optimal thresholds and thus mitigates the jumping effect.

In Fig.~\ref{fig:trace_res}, we also compare the performance of different policies using the real trace dataset \cite{cacheWorkload-OSDI20} when the staleness cost function is linear. The trace dataset collects around  $700$ billion user requests (each contains a timestamp, anonymized key, operation, etc.; see details in \cite{cacheWorkload-OSDI20}) from  $54$ Twemcache clusters, which are the  in-memory caching used by Twitter. 
In order to simplify the analysis and presentation, we focus on the user request arrival times at the cache of the $26$th Twemcache cluster, 
whose request arrival rate is about $0.4$.
The request arrival process is no longer Bernoulli. 
In Fig.~\ref{fig:trace_res}, we assume that the update cost is $p=25$ and show the performances of different policies using the trace dataset in the first $10^3$ requests. For the optimal offline policy, we still apply the dynamic programming to obtain the optimal update times. 
For the naive policy, we simply let its threshold equal the update cost of $25$. 
For the threshold-based policy, we obtain the optimal threshold ${\tau ^*} = 9$ by plugging $\lambda =0.4$ into Eq.~\eqref{eq:opt_threshold}. Similarly, we obtain the optimal period ${d^*} = 11$ for the periodic policy. We observe that the optimal threshold-based policy outperforms the other online policies even though the request arrival process is now non-Bernoulli. 
In Fig.~\ref{fig:square_trace_res}, we change the setting to update cost $p = 50$ and a quadratic staleness cost function. Clearly, the optimal threshold-based policy still outperforms the other online policies.

\section{Conclusion}\label{sec:conclusion}
In this paper, we considered a fundamental tradeoff between the data freshness and the update cost in a time-sensitive information-update system.
We provided useful guidelines for the design of update policies.
Assuming Bernoulli request arrival process, we also proposed a threshold-based update policy and proved its optimality.
Our simulations based on both synthetic data and real traces corroborated the theoretical results and showed that the optimal threshold-based policy outperforms the baseline policies.
For future work, one interesting direction would be to consider more general settings where users are interested in multiple contents.

\appendix

\subsection{Proof of Lemma \ref{lemma:reactive_policy_sufficient}}
\label{appendix:reactive_policy_proof}
\begin{proof}
For any given policy $\pi \in \Pi$, there are two cases: (i) $\pi \in \Pi^{R}$ and (ii) $\pi \in \Pi \setminus \Pi^{R}$. 

Case (i) is trivial as we can simply choose $\pi^{\prime} = \pi$. 
In Case (ii), we construct a reactive policy $\pi^{\prime} \in \Pi^{R}$ in the following manner. 
Consider any sample path with $N$ requests: 1) for all the updates performed by policy $\pi$ at some request arrival times, policy $\pi^{\prime}$ also has such updates; 2) for each update performed by policy $\pi$ that is not at any of the request arrival times, policy $\pi^{\prime}$ postpones the update to the time instant when the next request arrives. Clearly, policy $\pi^{\prime}$ constructed in the above manner is a valid reactive policy. 

Next, we want to prove that in Case (ii), policy $\pi^{\prime}$ achieves a total cost smaller than that of policy $\pi$ by induction.
Note that at time 1, the AoI under policy $\pi^{\prime}$ is no larger than that of policy $\pi$. (Although since the AoI at time 1 is equal under both policies, we want to use ``no larger than" instead of ``equal" here so that the argument can later be repeatedly applied.)
Let $u_k$ be the first time at which policy $\pi$ updates the data when no request arrives,  let $r_l$ be the arrival time of the first request that arrives after time $u_k$. 
Fig.~\ref{fig:reactive_policy} provides an illustration of such a scenario, where the solid black curve and the dashed red curve denote the AoI trajectories under policies $\pi$ and $\pi^{\prime}$, respectively.
We want to show that policy $\pi^{\prime}$ has a smaller total cost than that of policy $\pi$ during interval $[1,r_l]$.

\begin{figure}[!t]
    \centering
    \includegraphics[scale = 0.6]{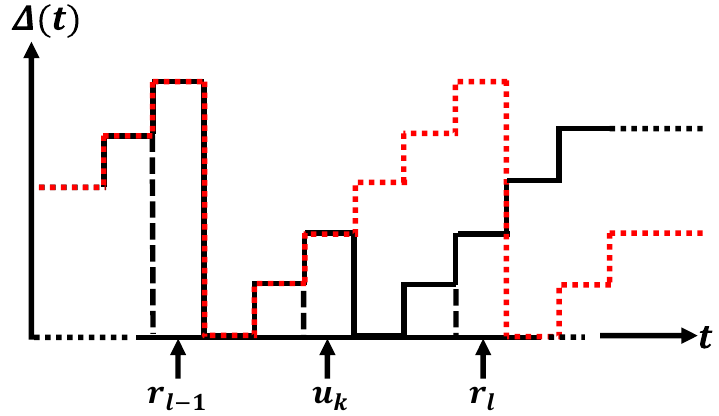}
    \vspace{-0.1cm}
    \caption{An illustration of the advantage of a constructed reactive policy $\pi^{\prime} \in \Pi^{R}$ over an arbitrary given non-reactive policy $\pi \in \Pi \backslash \Pi^{R}$, where the solid black curve and the dashed red curve denote the AoI trajectories under policies $\pi$ and $\pi^{\prime}$, respectively.}
    \label{fig:reactive_policy}
    \vspace{-0.45cm}
\end{figure}

We first compare the total cost during interval $[1, u_k)$ under both policies. Note that under policy $\pi$, all the updates performed before time $u_k$ are at request arrival times because $u_k$ is the first time at which policy $\pi$ updates the data when no request arrives.
Due to Step 1) of the above construction, policy $\pi^{\prime}$ must update the data in a way that is exactly the same as policy $\pi$ during interval $[1, u_k)$.
Hence, the total update cost during interval $[1, u_k)$ must be the same under both policies. Also, the AoI under policy $\pi^{\prime}$ must remain no larger than that under policy $\pi$ during interval $[1, u_k)$.
Hence, policy $\pi^{\prime}$ must have a total staleness cost no larger than that of policy $\pi$ during interval $[1, u_k)$. This implies that policy $\pi^{\prime}$ has a total cost (including both update cost and staleness cost) no larger than that of policy $\pi$ during interval $[1, u_k)$.

We now compare the total cost during the interval $[u_k, r_l]$ under both policies. Due to Step 2) of the construction of policy $\pi^{\prime}$, it does not update the data at time $u_k$ but postpones the update to time $r_l$. Then, under policy $\pi^{\prime}$ the total cost during interval $[u_k, r_l]$ is equal to $p$ because the only update is performed at time $r_l$ at which the only request arrives. On the other hand, policy $\pi$ has a total cost strictly larger than $p$ because it performs one update at time $u_k$ and needs to pay either a staleness cost at time $r_l$ or at least an additional update cost during interval $(u_k,r_l]$.
Hence, policy $\pi^{\prime}$ has a total cost smaller than that of policy $\pi$ during interval $[u_k,r_l]$.

With the above discussions, we show that policy $\pi^{\prime}$ achieves a total cost smaller than that of policy $\pi$ during interval $[1, r_l]$.

Since the AoI under policy $\pi^{\prime}$ becomes $1$ at time $r_{l+1}$, the AoI under policy $\pi^{\prime}$ is apparently no larger than that of policy $\pi$ at time $r_{l+1}$.
Then, we can view time $r_{l+1}$ as a new starting point and inductively apply the above argument.
This completes the proof.
\end{proof}

\subsection{Proof of Lemma~\ref{lemma:r+}}
\label{appendix:capped_reactive_policy_proof}
\begin{figure}[!t]
    \centering
    \includegraphics[scale = 0.6]{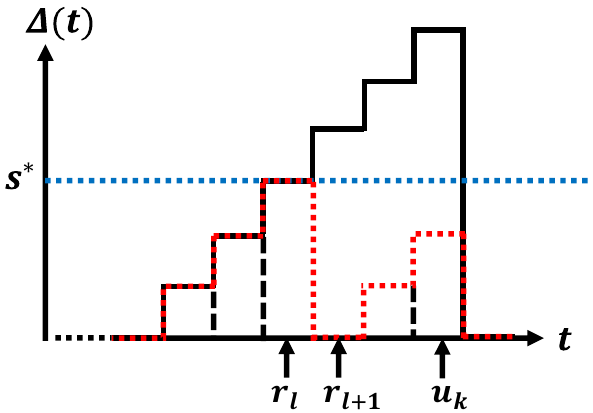}
    \caption{An illustration of the advantage of a constructed capped reactive policy $\pi^{\prime} \in \Pi^{R+}$ over an arbitrary given non-capped reactive policy $\pi\in\Pi^{R}\backslash\Pi^{R+}$, where the solid black curve and the dashed red curve denote the AoI trajectories under policies $\pi$ and $\pi^{\prime}$, respectively.}
    \label{fig:capped_reactive_policy}
    \vspace{-0.45cm}
\end{figure}
\begin{proof}
For any given policy $\pi \in \Pi^{R}$, there are two cases: (i) $\pi \in \Pi^{R+}$ and (ii) $\pi \in \Pi^{R} \setminus \Pi^{R+}$. 

Case (i) is trivial as we can simply choose $\pi^{\prime} = \pi$. 
In Case (ii), we construct a capped reactive policy $\pi^{\prime} \in \Pi^{R+}$ in the following manner.
Consider any sample path with $N$ requests:
1) for all the updates performed by policy $\pi$, policy $\pi^{\prime}$ also has such updates;
2) for every request $r_j$ such that $\Delta(r_j)\geq s^*$ (see definition of $s^*$ in Eq.~\eqref{eq:smallest_s}), policy $\pi^{\prime}$ updates the data at $r_j$ regardless of policy $\pi$'s update decision at $r_j$.
Clearly, policy $\pi^{\prime}$ constructed in the above manner is a valid capped reactive policy.

Next, we want to show that in Case (ii), policy $\pi^{\prime}$ achieves a total cost smaller than that of policy $\pi$ by induction.
Recall that the system starts with $\Delta(1)=1$.
Let $r_l$ be the first request arrival time at which $\Delta(r_j) \geq s^*$ but policy $\pi$ does not update the data.
Then, the AoI trajectories before time $r_l$ are exactly the same under policies $\pi^{\prime}$ and $\pi$. 
Due to Step 2) of the above construction, policy $\pi^{\prime}$ must update the data at $r_l$ because $\Delta(r_j) \geq s^*$.
Let $u_k$ be the first update performed after $r_l$ under policy $\pi$.
Then, policy $\pi^{\prime}$ must also update the data at time $u_k$ due to Step 1) of the above construction.
Fig.~\ref{fig:capped_reactive_policy} provides an illustration of such a scenario, where the solid black curve and the dashed red curve denote the AoI trajectories under policies $\pi$ and $\pi^{\prime}$, respectively. It is easy to see that policy $\pi^{\prime}$ has a smaller AoI than that under policy $\pi$ during interval $(r_l, u_k)$.

We now compare the total cost during interval $[1, u_k]$ under both policies.
First, note that the AoI trajectories before time $r_l$ are exactly the same under policies $\pi^{\prime}$ and $\pi$, i.e., both policies have the same cost during interval $[1,r_l)$. Then, it is easy to see that at time $r_l$, policy $\pi^{\prime}$ has an update cost $p$ but no staleness cost, and policy $\pi$ has a staleness cost $\Delta(r_j)$ but no update cost. Hence, policy $\pi^{\prime}$ achieves a smaller total cost than that of policy $\pi$ at time $r_l$ because $\Delta(r_j) \geq s^*$.
Also, note that policy $\pi^{\prime}$ has a smaller AoI than that of policy $\pi$ during interval $(r_l, u_k)$ and that there is no update under both policies during interval $(r_l, u_k)$. Hence, policy $\pi^{\prime}$ achieves a smaller total cost than that of policy $\pi$ during interval $(r_l, u_k)$. At time $u_k$, both policies have an update cost $p$ but no staleness cost since neither of them updates the data. Combining the above discussions, we show that policy $\pi^{\prime}$ achieves a smaller total cost than that of policy $\pi$ during interval $[1, u_k]$.

Since $\Delta(u_k)$ drops to $0$  at time $u_{k}$ under both policies, we can view time $(u_{k}+1)$ as a new starting point and repeatedly apply the above argument. This completes the proof.
\end{proof}

\subsection{Proof of Theorem \ref{thm:exi_opt_avg_mdp}}
\label{appendix:theorem_1}
\begin{proof}
Our proof includes two steps: $1$) We study a discounted MDP and show that the optimal value function of the discounted MDP is non-decreasing in the initial state $s$; 
$2$) Based on the optimal value function of the discounted MDP, we derive the Bellman equation of the average expected cost  and show that it has a threshold-based structure, where the threshold is based on the AoI.
This implies that there exists an optimal threshold-based stationary capped reactive  policy for the average expected cost.

\textbf{Step 1):} In general, the derivation and properties of the Bellman equation of the average expected cost are not easy to obtain, and we usually rely on the study of the discounted MDP to get some insights towards the design of an optimal policy~\cite{ross2014introduction}.

The expected total $\alpha$-discounted cost of a capped reactive policy $\pi  \in {\Pi ^{{R+ }}}$ is defined as  
\begin{equation}
	C_\alpha ^\pi (s) \triangleq {\mathbb{E}_\pi }\left[ {\sum\limits_{n = 1}^\infty  {{\alpha ^{n-1}}c({s_n},{a_n})} |{s_1} = s} \right], 
\end{equation}
where $\ 0 \le \alpha  < 1$ is the discount factor.
Here, $C_\alpha ^\pi (s)$ is well defined, given that for any $n$, we have ${\mathbb{E}_\pi }[c({s_n},{a_n})|{s_1} = s] \leq p$ under a capped reactive policy $\pi$, and thus, we have
\vspace{-0.05cm}
\begin{equation}
    C_\alpha ^\pi (s) \le  \sum\limits_{n = 1}^\infty {{\alpha ^{n - 1}}p}  = \frac{p}{{1 - \alpha }}.
    \vspace{-0.05cm}
\end{equation}

Let ${C_\alpha }\left( s \right) \triangleq \mathop {\min }\limits_\pi  C_\alpha ^\pi (s)$ be the optimal value function. 
Then, we can obtain the Bellman equation of the $\alpha$-discounted MDP with ${C_\alpha}(s)$ \cite{ross2014introduction}, which is
\vspace{-0.05cm}
\begin{equation}
\begin{array}{*{20}{l}}
{{C_\alpha}(s) = \mathop {\min }\limits_{a \in \mathcal{A}_s} \left\{ {c(s,a) + \alpha \sum\limits_{z\in S}   p(z\mid s,a){C_\alpha}(z)} \right\}}.
\end{array}
\label{eq:gel_dis_opt_eq}
\vspace{-0.05cm}
\end{equation}
The Bellman equation Eq.~\eqref{eq:gel_dis_opt_eq}
states that the value of the initial state $s$ (i.e., ${C_\alpha }(s)$) equals the expected return of the best action, which is
the discounted expected value of the next state (i.e., $\alpha \sum\nolimits_{z \in S} {p(z\mid s,a){C_\alpha }(z)} $), plus the immediate cost along the way (i.e., ${c(s,a)}$).  
In the following, we show that ${C_\alpha}(s)$ is non-decreasing in $s$. This property enables us to show that the Bellman equation of the average cost (i.e., Lemma~\ref{lemma:bellman_avg}) has a threshold-based structure.

\begin{lemma}
\vspace{-0.1cm}
The optimal value function  ${C_\alpha}(s)$ is non-decreasing in the initial state $s$.
\label{lemma:non-decreasing}
\vspace{-0.1cm}
\end{lemma}

We provide the proof of Lemma~\ref{lemma:non-decreasing} in Appendix~\ref{proof:non-decreasing} and explain the key ideas in the following. 
The goal is to construct a sequence $\{ {C_{\alpha ,n}}\left( s \right)\} $ that is non-decreasing in $s$ for any $n$, where $C_{\alpha,n}(s)$ is the minimal expected discounted cost in an $n$-stages problem. Then, we show that ${C_\alpha }(s) = \mathop {\lim }\limits_{n \to \infty } {C_{\alpha ,n}}(s)$, which implies that ${C_\alpha }(s)$ is also non-decreasing in $s$. 

\textbf{Step 2):} With the optimal value function of the $\alpha$-discounted MDP (i.e., ${C_\alpha }(s)$), we can derive the Bellman equation of the average cost as follows.

\vspace{-0.3cm}
\begin{lemma}
\label{lemma:bellman_avg}
Let $h(s) \triangleq \mathop {\lim }\limits_{\alpha  \to 1} [{C_\alpha }(s) - {C_\alpha }(1)]$ and $g \triangleq \mathop {\lim }\limits_{\alpha  \to 1} (1 - \alpha ){C_\alpha }(1)$. Then, the Bellman equation of the average cost is given by the following:
\vspace{-0.1cm}
\begin{equation}
\begin{array}{*{20}{l}}
{h(s) + g} =\\
{\left\{ {\begin{array}{*{20}{ll}}
{p  + \sum\limits_{z = 1}^\infty  {{{(1 - \lambda )}^{z - 1}}} \lambda h(z),\;{\rm{if}}\;s \ge {\Delta ^*};}\\
{\min \left\{ {p + \sum\limits_{z = 1}^\infty  {{{(1 - \lambda )}^{z - 1}}} \lambda h(z),} \right.}\\
{\left. {f(s) + \sum\limits_{z = s + 1 }^\infty  {{{(1 - \lambda )}^{z - s - 1}}} \lambda h(z)} \right\},\;{\rm{if}}\;s < {\Delta ^*}.}
\end{array}} \right.}
\end{array}
\label{eq:gel_avg_opt_eq}
\end{equation}
\vspace{-0.7cm}
\end{lemma}

We provide the proof of Lemma~\ref{lemma:bellman_avg} in Appendix~\ref{proof:Bellman_Avg_Cost} and explain the key ideas in the following.
First, given the definitions of $h(s)$ and $g$, we show that Eq.~\eqref{eq:gel_avg_opt_eq} does hold.  
To this end, we define ${h_\alpha }(s) \triangleq {C_\alpha }(s) - {C_\alpha }(1)$ and substitute ${h_\alpha }(s)$ into the Bellman equation of the $\alpha$-discounted MDP Eq.~(\ref{eq:gel_dis_opt_eq}). Then, we prove that we can find a sequence $\{ {\alpha _m}\}~\to~1$ such that $\mathop {\lim }\limits_{m \to \infty }  {h_{{\alpha _m}}}(s)  = h(s)$ for any $s$ and $\mathop {\lim }\limits_{m \to \infty }  (1 - {\alpha _m}){C_{{\alpha _m}}}(1)  = g$. Taking the limit $m \to \infty $ on both sides of the Bellman equation of ${h_\alpha }(s)$, we obtain Eq.~\eqref{eq:gel_avg_opt_eq}. 
Second, we show that Eq.~\eqref{eq:gel_avg_opt_eq} is the Bellman equation for the average expected cost. This can be done by applying the same techniques used in \cite[Chapter V, Theorem $2.1$]{ross2014introduction}.

Next, we show that the Bellman equation Eq.~\eqref{eq:gel_avg_opt_eq} has a threshold structure, which guides us to find the optimal threshold-based stationary capped  reactive policy.

Assume that the current state is $s$.
Based on the Bellman equation Eq.~\eqref{eq:gel_avg_opt_eq}, it is optimal to update when $s \ge {\Delta ^*}$; and when $s < {\Delta ^*}$, it is optimal to update if  
\vspace{-0.1cm}
\begin{equation}
\begin{aligned}
    f(s) + \alpha \sum\limits_{z = s + 1}^\infty  {{{(1 - \lambda )}^{z - s - 1}}} \lambda h(z) \ge  \\ p +\alpha \sum\limits_{z = 1}^\infty  {{{(1 - \lambda )}^{z - 1}}} \lambda h(z),
\end{aligned}
\end{equation}
\vspace{-0.1cm}
where the right hand side is a constant. 
It is easy to check that $h(s)$ is non-decreasing in $s$ given that ${C_\alpha }\left( s \right)$ is non-decreasing in $s$. 
Hence, we can find $s^*$ as follows: 
\vspace{-0.1cm}
\begin{equation}
\label{eq:min_i}
    \begin{aligned}
    & s^*  { \triangleq \min \left\{ {s: f(s)  + \alpha \sum\limits_{z = s + 1}^\infty  {{{(1 - \lambda )}^{z - s - 1}}} \lambda h(z)}  \ge  \right.}\\ 
	& {\left. {p + \alpha \sum\limits_{z = 1}^\infty  {{{(1 - \lambda )}^{z - 1}}} \lambda h(z)} \right\}}.
    \end{aligned}
    \vspace{-0.1cm}
\end{equation}
Set $s^*={\Delta^*}$ when $s^* \geq \Delta^*$. 
Now consider a capped reactive policy $\pi^*  \in {\Pi ^{R + }}$:
upon receiving a request, policy $\pi^*$ updates the data if the current state is no smaller than $s^*$; otherwise, it does not update the data and replies with the current local data.
Clearly, policy $\pi^*$  is a stationary threshold-based policy. Besides, policy $\pi^*$ selects the action that minimizes the right hand side of the Bellman equation Eq.~\eqref{eq:gel_avg_opt_eq} for any state. Thus, it is an optimal policy~\cite[Chapter V, Theorem $2.1$]{ross2014introduction}. 
\end{proof}

\subsection{Proof of Lemma~\ref{lemma:non-decreasing}}

\begin{proof}
\label{proof:non-decreasing}
Our proof idea is to construct a sequence $\{ {C_{\alpha ,n}}\left( s \right)\} $ that is non-decreasing in $s$ for any $n$, where $C_{\alpha,n}(s)$ is the minimal expected discounted cost in an $n$-stages problem. Then, we show that ${C_\alpha }(s) = \mathop {\lim }\limits_{n \to \infty } {C_{\alpha ,n}}(s)$, which implies that ${C_\alpha }(s)$ is also non-decreasing in $s$.

First, we show how to construct the sequence $\{ {C_{\alpha ,n}}\left( s \right)\} $.
Consider an $n$-stage problem of our $\alpha$-discounted MDP. 
Denote the minimal expected discounted cost of this $n$-stage problem by
\vspace{-0.2cm}
\begin{equation*}
	\begin{aligned}
	 & {{C_{\alpha ,n}}(s)  \triangleq \mathop {\min }\limits_{a \in {A_s}} \left\{ {c(s,a) + \alpha \sum\limits_{z \in S} p (z\mid s,a){C_{\alpha ,n - 1}}(z)} \right\}} \\
	& = \left\{ {\begin{array}{*{20}{l}}
    {p + \alpha \sum\limits_{z = 1}^\infty  {{{(1 - \lambda )}^{z - 1}}} \lambda {C_{\alpha ,n-1}}(z),\;{\rm{if}}\;s \ge {\Delta ^*};}\\
   {\min \left\{ {p + \alpha \sum\limits_{z = 1}^\infty  {{{(1 - \lambda )}^{z - 1}}} \lambda {C_{\alpha ,n-1}}(z),} \right.}\\
   {\left. {f(s) + \alpha \sum\limits_{z = s + 1}^\infty  {{{(1 - \lambda )}^{z - s - 1}}} \lambda {C_{\alpha,n-1} }(z)} \right\},\;{\rm{if}}\;s < {\Delta ^*},}
   \end{array}} \right.
	\end{aligned}
 \vspace{-0.2cm}
\end{equation*}
where the terminal cost is ${C_{\alpha ,1}}\left( s \right) \triangleq \min \{ p,f(s)\}$.

Then, we prove by induction that our constructed sequence $\{ {C_{\alpha ,n}}\left( s \right)\} $ is non-decreasing in $s$ for any $n$. Obviously, ${C_{\alpha, 1}}(s)$ is non-decreasing in $s$. 
We assume that ${{C_{\alpha, n - 1}}(s)}$ is non-decreasing in $s$. 
Next, we show that ${{C_{\alpha, n}}(s)}$ is non-decreasing in $s$. 
When $s \ge {\Delta ^*}$, $C_{\alpha, n}(s)$ is a constant and is independent of $s$. 
On the other hand, when $s < {\Delta ^*}$, to better present our discussion, we denote 
\vspace{-0.1cm}
\begin{equation*}
    c \triangleq {p + \alpha \sum\limits_{z = 1}^\infty  {{{(1 - \lambda )}^{z - 1}}} \lambda {C_{\alpha ,n-1}}(z)}
    \vspace{-0.1cm}
\end{equation*}
and 
\begin{align*}
    {C^{\prime}_{\alpha,n-1} }(s) \triangleq& {f(s) + \alpha \sum\limits_{z = s + 1}^\infty  {{{(1 - \lambda )}^{z - s - 1}}} \lambda {C_{\alpha,n-1} }(z)}\\
    =&f(s) + \alpha \sum\limits_{k= 1}^\infty  {{{(1 - \lambda )}^{k - 1}}} \lambda {C_{\alpha,n-1} }(s+k),
\end{align*}
where the last step follows by setting $k=z-s$.
We want to show that ${C^{\prime}_{\alpha,n-1} }(s)$ is non-decreasing in $s$. To see this, consider two states: $ {j} \geq {i} > 0$. Then, we have
\begin{equation*}
    \begin{aligned}
        & {C^{\prime}_{\alpha ,n}}(j) - {C^{\prime}_{\alpha ,n}}(i)  \\
        & =  f(j) - f(i)  \\ & +  \alpha \sum\limits_{k = 1}^\infty  {{{(1 - \lambda )}^{k - 1}}} \lambda ({C_{\alpha ,n - 1}}({j} + k) - {C_{\alpha ,n - 1}}({i} + k)).
    \end{aligned}
\end{equation*}
Since $f(s)$ is non-decreasing in $s$, we have $f(j) \geq f(i)$. By inductive hypothesis that ${{C_{\alpha, n - 1}}(s)}$ is non-decreasing in $s$,
we have ${C_{\alpha ,n - 1}}({j} + k) \geq {C_{\alpha ,n - 1}}({i} + k)$ for any $k > 0$. Therefore,  ${C^{\prime}_{\alpha ,n}}(j) - {C^{\prime}_{\alpha ,n}}(i) \geq 0$, i.e., ${C_{\alpha ,n}}(s)$ is non-decreasing in $s$. This, along with $c$ being a constant, implies that $C_{\alpha, n}(s) = \min \{c,C^{\prime}_{\alpha,n-1}(s)\}$ is also non-decreasing in $s$ when $s < {\Delta ^*}$. 
Till this point, we have shown that ${C_{\alpha ,n}}(s)$ is non-decreasing in $s$ when $s < {\Delta ^*}$ and  when $s \ge {\Delta ^*}$, respectively. 
In addition, since ${C_{\alpha ,n}}(s)$ achieves a smaller value when $s < {\Delta ^*}$ (i.e., $\min \left\{ {c,C_{\alpha ,n - 1}^\prime (s)} \right\}$) compared to the case of $s \ge {\Delta ^*}$ (i.e., $c$), this implies that ${C_{\alpha ,n}}\left( s \right)$ is non-decreasing in $s$ for any $n$.

Finally, given any non-negative integer $s$, by the definition of  $C_{\alpha,n}(s)$ and the fact that the cost in each stage is non-negative, we know that $C_{\alpha,n}(s)$ is non-descreasing in $n$. In addition, 
${C_{\alpha ,n}}(s)$ is bounded. Indeed, the cost in each stage is bounded by the update cost $p$ under our considered policy. As such, 
${C_{\alpha ,n}}(s) \le \sum\nolimits_{i = 1}^n {{\alpha ^{i - 1}}p = (1 - {\alpha ^n})p/(1 - \alpha ) \le } p/(1 - \alpha )$.
Therefore, by monotone convergence theorem, we have
\vspace{-0.2cm}
\begin{equation}
    \mathop {\lim }\limits_{n \to \infty } {C_{\alpha ,n}}\left( s \right) = {C_\alpha }\left( s \right),
    \label{eq:convergence}
    \vspace{-0.2cm}
\end{equation}
holding for any state $s=0,1,2,\dots$.
It still remains to prove  that ${C_\alpha }\left( s \right)$ is non-decreasing in $s$. 
We prove this by contradiction. Suppose that ${C_\alpha }\left( s \right)$ is decreasing in $s$, i.e., there exists $i,j \in \mathcal{S}$ such that $i < j$ and ${C_\alpha }\left( i \right) > {C_\alpha }\left( j \right)$.
Let 
\vspace{-0.2cm}
\begin{equation}
    r = {C_\alpha }\left( i \right) - {C_\alpha }\left( j \right) > 0.
    \vspace{-0.2cm}
\end{equation}
Note that ${C_{\alpha ,n}}\left( s \right)$ converges to ${C_\alpha }\left( s \right)$ pointwise for any $s$,  there exists positive integers ${N_1}$ and $N_2$ such that for any $n > \max\{N_1,N_2\}$,  
\vspace{-0.2cm}
\begin{align}
    \left| {{C_\alpha }(i) - {C_{\alpha ,n}}(i)} \right| < \frac{r}{3}\\
   \text{and } \left| {{C_\alpha }(j) - {C_{\alpha ,n}}(j)} \right| < \frac{r}{3},
   \vspace{-0.2cm}
\end{align}
which are equivalent to
\vspace{-0.2cm}
\begin{equation}
     - \frac{r}{3} < {C_\alpha }(i) - {C_{\alpha ,n}}(i) < \frac{r}{3}
     \label{eq:C_i}
\end{equation}
\vspace{-0.2cm}
\begin{equation}
   \text{and }  - \frac{r}{3} < {C_\alpha }(j) - {C_{\alpha ,n}}(j) < \frac{r}{3}.
     \label{eq:C_j}
     \vspace{-0.2cm}
\end{equation}
By subtracting Eq.~\eqref{eq:C_i} from Eq.~\eqref{eq:C_j}, we have  
\begin{equation}
    - \frac{{2r}}{3} < ({C_\alpha }(j) - {C_{\alpha ,n}}(j)) - ({C_\alpha }(i) - {C_{\alpha ,n}}(i)) < \frac{{2r}}{3},
\end{equation}
Note that 
\begin{equation}
    \begin{aligned}
        & ({C_\alpha }(j)  - {C_{\alpha ,n}}(j)) - ({C_\alpha }(i) - {C_{\alpha ,n}}(i)) \\
        & = ({C_{\alpha ,n}}(i) - {C_{\alpha ,n}}(j)) - ({C_\alpha }(i) - {C_\alpha }(j)) \\
        & = ({C_{\alpha ,n}}(i) - {C_{\alpha ,n}}(j)) - r.
    \end{aligned}
\end{equation}
Then, we have 
\begin{equation}
    {C_{\alpha ,n}}(i) - {C_{\alpha ,n}}(j) > \frac{r}{3} > 0,
\end{equation}
which contradicts the fact that ${C_{\alpha ,n}}\left( s \right)$ is non-decreasing in $s$ for any $n$. Therefore, ${C_\alpha }\left( s \right)$ must be non-decreasing in $s$.
\end{proof}

\subsection{Proof of Lemma~\ref{lemma:bellman_avg}}
\label{proof:Bellman_Avg_Cost}

\begin{proof}
The proof consists two steps: 1) given the definitions of $h(s)$ and $g$, we show that Eq.~\eqref{eq:gel_avg_opt_eq} dose hold; 2) we show that Eq.~\eqref{eq:gel_avg_opt_eq} is the Bellman equation for the average cost.

\textbf{Step 1):}
As we can see that the Bellman equation Eq.~\eqref{eq:gel_avg_opt_eq} consists of two cases: $s \geq {\Delta ^*}$ and $s < {\Delta ^*}$. In this proof, we only prove the case of $s < {\Delta ^*}$, and a similar proof can also be applied to the case of $s \geq {\Delta ^*}$.

Define
\vspace{-0.1cm}
\begin{equation}
	{h_\alpha }(s) \triangleq {C_\alpha }(s) - {C_\alpha }(1),
 \vspace{-0.1cm}
\end{equation}
and substituting ${C_\alpha }(s)$ into the Bellman equation of the $\alpha$-discounted MDP Eq.~(\ref{eq:gel_dis_opt_eq}) gives
\vspace{-0.1cm}
\begin{equation}
\begin{array}{*{20}{l}}
{{h_\alpha }(s) + (1 - \alpha ){C_\alpha }(1)}\\
{ = \min \left\{ {p + \alpha \sum\limits_{z = 1}^\infty  {{{(1 - \lambda )}^{z - 1}}} \lambda {h_\alpha }(z),} \right.}\\
{\left. {s  + \alpha \sum\limits_{z = s + 1}^\infty  {{{(1 - \lambda )}^{z - s - 1}}} \lambda {h_\alpha }(z)} \right\}.}
\end{array}
\label{eq:h_dis_opt_eq}
\vspace{-0.1cm}
\end{equation}
Our idea is to find a sequence $\{ {\alpha _m}\}~\to~1$ (where $\ 0 \le \alpha_m  < 1$ for any $m$) such that $\mathop {\lim }\limits_{m \to \infty }  {h_{{\alpha _m}}}(s)  = h(s)$ for any $s$ and $\mathop {\lim }\limits_{m \to \infty }  (1 - {\alpha _m}){C_{{\alpha _m}}}(1)  =g$.

From Eq.~(\ref{eq:gel_dis_opt_eq}), we know that
\begin{equation*}
    \begin{aligned}
        {h_\alpha }(s) &= {C_\alpha }(s) - {C_\alpha }(1) \\
		& \le p  + \alpha \sum\limits_{z = 1}^\infty  {{{(1 - \lambda )}^{z - 1}}} \lambda {C_\alpha }(z) - {C_\alpha }(1) \\
		& = p  + \alpha \sum\limits_{z = 2}^\infty  {{{(1 - \lambda )}^{z - 1}}} \lambda {C_\alpha }(z) + (\alpha \lambda  - 1){C_\alpha }(1) \\
		& \le p + \alpha \sum\limits_{z = 2}^\infty  {{{(1 - \lambda )}^{z - 1}}} \lambda {C_\alpha }(z),
    \end{aligned}
\end{equation*}
where the last quantity is a constant and we denote it as $M \triangleq p + \alpha \sum\nolimits_{z = 2}^\infty  {{{(1 - \lambda )}^{z - 1}}\lambda {C_\alpha }(z)} $.  This immediately gives $\left| {{h_\alpha }(s)} \right| \le M$ because ${C_\alpha }\left( s \right)$ is non-decreasing in $s$ and thus ${{h_\alpha }(s)} \geq 0$, which also implies that ${h_\alpha }(s)$ is uniformly bounded for all $s$ and $\alpha$. 
Because every bounded sequence contains a convergent subsequence \cite[Bolzano–Weierstrass Theorem]{royden1988real}, we can find a sequence $\{ \alpha _m^{(1)}\} $ with $\mathop {\lim }\limits_{m \to \infty }  \alpha _m^{(1)}   = 1$ such that $\mathop {\lim }\limits_{m \to \infty } {h_{\alpha _m^{(1)}}}(1) \triangleq h(1)$ exists. Furthermore, because ${h_{\alpha _m^{(1)}}}(2)$ is also bounded, we can find a subsequence $\{ \alpha _m^{(2)}\}  \subseteq \{ \alpha _m^{(1)}\} $ such that $\mathop {\lim }\limits_{m \to \infty } {h_{\alpha _m^{(2)}}}(2) \triangleq h(2)$ exists. Similarly, we can continue this argument for $h(3)$ and so on. 
Finally, let $\{ {\alpha _m}\}  = \{ \alpha _m^{(m)}\} $, and when $m \to \infty $, it follows that $\mathop {\lim }\limits_{m \to \infty } {h_{{\alpha _m}}}(1) = h(1),\mathop {\lim }\limits_{m \to \infty } {h_{{\alpha _m}}}(2) = h(2),$ and so on. Therefore, $\mathop {\lim }\limits_{m \to \infty } {h_{{\alpha _m}}}(s) = h(s)$ for any $s$.

Under a capped reactive policy $\pi$, we have ${\mathbb{E}_\pi }[c({s_n},{a_n})|{s_1} = s] \leq p$. Then, for any ${{\alpha _m}}$, we obtain
\begin{equation}
	\begin{aligned}
		{C_{{\alpha _m}}}(1) &  = {\mathbb{E}_{{\pi}}}\left[ {\sum\limits_{n = 1}^\infty  {\alpha _m^{n - 1}c({s_n},{Y_n})} |{s_1} = 1} \right] \\
		& \le {\sum\limits_{n = 1}^\infty  {\alpha _m^{n - 1}} p } 
		 = \frac{p}{{1 - {\alpha _m}}},
	\end{aligned}
\end{equation}
which indicates ${C_{{\alpha _m}}}(1)$ is bounded.  This also implies that both of ${C_{{\alpha _m}}}(1)$ and $(1 - {\alpha _m}){C_{{\alpha _m}}}(1)$ are uniformly bounded for all ${\alpha _m}$. 
Hence, there exists a subsequence $\{ {\alpha _{\bar m}}\}  \subseteq \{ {\alpha _m}\} $ such that $\mathop {\lim }\limits_{{\bar m} \to \infty } (1 - {\alpha _{\bar m}}){C_{{\alpha _{\bar m}}}}(1) \triangleq g$ exists. Note that $\mathop {\lim }\limits_{{\bar m} \to \infty } {h_{{\alpha _{\bar m}}}}(s) = h(s)$ also hold for any $s$.

Plug ${\alpha _{\bar m}}$ into the Bellman equation Eq.~\eqref{eq:h_dis_opt_eq}, we have 
\begin{equation}
\begin{array}{*{20}{l}}
{{h_{\alpha _{\bar m}} }(s) + (1 - {\alpha _{\bar m}} ){C_{\alpha _{\bar m}} }(1)} = \\
{ \min \left\{ {p + 1 + {\alpha _{\bar m}} \sum\limits_{j = 1}^\infty  {{{(1 - \lambda)}^{j - 1}}} \lambda {h_{\alpha _{\bar m}} }(j),} \right.}\\
{\left. {s + 1 + {\alpha _{\bar m}} \sum\limits_{j = s + 1}^\infty  {{{(1 - \lambda )}^{j - s - 1}}} \lambda {h_{\alpha _{\bar m}}}(j)} \right\}}.
\end{array}
\label{eq:h_dis_opt_eq_m}
\end{equation}
By the Lebesgue's Bounded Convergence Theorem \cite{royden1988real} and the boundedness of ${h_{{\alpha _{\bar m}}}}(s)$, it follows that
\begin{equation}
\mathop {\lim }\limits_{\bar m \to \infty } \sum\limits_{j = k}^\infty  {{{(1 - \lambda )}^{j - 1}}} \lambda {h_{{\alpha _{\bar m}}}}(j) = \sum\limits_{j = k}^\infty  {{{(1 - \lambda )}^{j - 1}}} \lambda h(j)
\end{equation}
for any $k$.
Finally, Eq.~\eqref{eq:gel_avg_opt_eq} holds by taking limit as  $\bar m \to \infty $ at both sides of Eq.~\eqref{eq:h_dis_opt_eq_m}.

\textbf{Step 2):} 
To prove that Eq.~\eqref{eq:gel_avg_opt_eq} is the Bellman equation for the average expected cost, we apply the same techniques used in \cite[Chapter V, Theorem $2.1$]{ross2014introduction}. 
For the ease of expression, we rewrite Eq.~\eqref{eq:gel_avg_opt_eq} as a compact form, i.e., 
\begin{equation}
    {h(s) + g = \mathop {\min }\limits_{a \in {A_s}} \left\{ {c(s,a) + \sum\limits_{z \in S} p (z\mid s,a)h(z)} \right\}}.
    \label{eq:gel_avg_opt_eq2}
\end{equation}

First of all, we claim that $g$ is the optimal average expected cost, i.e.,
\begin{equation}
    g = \mathop {\min }\limits_{\pi  \in {\Pi ^{R + }}} \mathop {\lim }\limits_{N \to \infty } \mathbb{E}_\pi\left[ {\sum\limits_{n = 1}^N {c({s_n},{a_n})} }| s_1 = s \right] / N.
\end{equation}
To see this, let ${H_n} \triangleq ({s_1},{a_1}, \ldots ,{s_n},{a_n})$ denote the history of the process up to the decision epoch $n$. By the iterated expectation, for any decision epoch $i$ under policy $\pi$, we have 
\begin{equation}
    \mathbb{E}_\pi[h({s_i})] = \mathbb{E}_\pi[\mathbb{E}_\pi[h({s_i})|{H_{i - 1}}]],
\end{equation}
which gives 
\begin{equation}
\label{eq:exp_iterated_exp}
    {\mathbb{E}_\pi }\left[ {\sum\limits_{i = 1}^N {[h({s_i}) - {\mathbb{E}_\pi }[h({s_i})|{H_{i - 1}}]]} } \right] = 0.
\end{equation}
Based on the definition of ${\mathbb{E}_\pi }[h({s_i})|{H_{i - 1}}]$, we have
\begin{equation}
\label{eq:exp_hs_history}
    \begin{aligned}
        &{\mathbb{E}_\pi }[h({s_i})|{H_{i - 1}}] \\
        &= \sum\limits_{z \in S}  {p(z|{s_{i - 1}},{a_{i - 1}})h(z)} \\ 
        &= c({s_{i - 1}},{a_{i - 1}}) + \sum\limits_{z \in S} {p(z|{s_{i - 1}},{a_{i - 1}})h(z)}  - c({s_{i - 1}},{a_{i - 1}}) \\
        & \buildrel (a) \over \geq \mathop {\min }\limits_{{a \in {A_s}}} \left[ {c({s_{i - 1}},{a}) + \sum\limits_{z \in S} {p(z|{s_{i - 1}},{a})h(z)} } \right] 
        - c({s_{i - 1}},{a_{i - 1}}) \\
        &  \buildrel (b) \over = g + h({s_{i - 1}}) - c({s_{i - 1}},{a_{i - 1}}),
    \end{aligned}
\end{equation}
where (a) becomes equation under the optimal policy $\pi^*$, 
and (b) holds because of Eq.~\eqref{eq:gel_avg_opt_eq2}.

Taking Eq.~\eqref{eq:exp_hs_history} into Eq.~\eqref{eq:exp_iterated_exp}, we obtain 
\begin{equation}
    0 \leq {\mathbb{E}_\pi }\left[ {\sum\limits_{i = 1}^N {[h({s_i}) - g - h({s_{i - 1}}) + c({s_{i - 1}},{a_{i - 1}})]} } \right],
\end{equation}
which can be rewritten as 
\begin{equation}
    g \leq \frac{{{\mathbb{E}_\pi }[h({s_n})]}}{N} - \frac{{{\mathbb{E}_\pi }[h({s_1})]}}{N} + \frac{{{\mathbb{E}_\pi }\left[ {\sum\limits_{i = 1}^N {c({s_{i - 1}},{a_{i - 1}})} } \right]}}{N},
\end{equation}
where the inequality becomes equality under the optimal policy $\pi^*$. 
Taking the limit as $N \to \infty $ and using the fact that $h(s)$ is bounded for any $s$, we have 
\begin{equation}
    g \leq \mathop {\lim }\limits_{N \to \infty } \mathbb{E}_\pi\left[ {\sum\limits_{n = 1}^N {c({s_n},{a_n})} }| s_1 = s \right] / N,
\end{equation}
with equality for the optimal policy $\pi^*$ and any initial value $s_1$. That is, $g$ is the optimal average expected cost. This also implies that Eq.~\eqref{eq:gel_avg_opt_eq} is the Optimality Equation since any $h(s)$ and $g$ satisfying Eq.~\eqref{eq:gel_avg_opt_eq} result in the optimal average expected cost $g$. 
This completes the proof. 
\end{proof}

\subsection{Proof of Theorem~\ref{theorem:thresholdperformance}}
\label{appendix:theorem_2}
\begin{proof}
We start with some additional notations. For the $k$-th update interval $[{u_{k - 1}^\pi} + 1,{u_k^\pi }]$, we use $N_{k}$ and $C_{k}$ to denote the number of requests that arrive in $[{u_{k - 1}^\pi} + 1,{u_k^\pi }]$ and the total cost serving these $N_{k}$ requests, respectively.  

Since the request arrival process is Bernoulli, under a threshold-based update policy, the lengths of update intervals are \textit{i.i.d.}, so the update process is a renewal process. 
By the ergodicity of the process,  the average cost can be rewritten as 
\begin{equation}
\label{eq:elim}
    {\bar{C}^{\pi(\tau)} } = \mathbb{E}[{C_{k}}]/\mathbb{E}[{N_{k}}]. 
\end{equation}

To calculate $\mathbb{E}[N_{k}]$, we consider the requests that arrive in $[{u_{k - 1}^\pi} + 1,{u_k^\pi }]$. Apparently, there is only one request in $[{u_{k - 1}^\pi} + \tau ,{u_k^\pi }]$, which arrives exactly at ${u_k^\pi }$ because of the threshold-based policy. Besides, the expected number of requests arriving in $[{u_{k - 1}^\pi} + 1,{u_{k - 1}^\pi} + \tau  - 1]$ is $\lambda (\tau  - 1)$ according to the Bernoulli process. Therefore, we have 
\begin{equation}\label{eq:exp_num_requrests}
    \mathbb{E}[{N_{k}}] = \lambda (\tau  - 1)  + 1.
\end{equation}

The total cost in an update interval is composed  of an update cost and  some staleness costs, i.e., 
\vspace{-0.1cm}
\begin{equation}
    {C_{k}} = p + \sum\limits_{n = 1}^{{N_{k}}} {f(\Delta ({r_n})){\mathds{1}_{\{ {N_{k}} > 0\} }}},
\end{equation}
where ${\mathds{1}_{\{  \cdot \} }}$ is the indicator function. 
Here, we slightly abuse the notation of $n$ and use it to denote the index of requests arriving in $[{u_{k - 1}^\pi} + 1,{u_k^\pi }]$ (i.e., ${r_n}$ is the arrival time of the $n$-th request in $[{u_{k - 1}^\pi} + 1,{u_k^\pi }]$).
The staleness costs can be rewritten as
\begin{equation}
    \begin{aligned}
        \sum\limits_{n = 1}^{{N_{k}}} {f(\Delta ({r_n})){\mathds{1}_{\{ {N_{k}} > 0\} }}}  &= \sum\limits_{n = 1}^\infty  {f(\Delta ({r_n})){\mathds{1}_{\{ n \le {N_{k}}\} }}} \\ &= \sum\limits_{n = 1}^\infty  {f(\Delta ({r_n})){\mathds{1}_{\{\Delta ({r_n}) \le \tau - 1 \} }}}.
    \end{aligned}
\end{equation}
For the expected staleness cost of the $n$-th arrival of the requests, we have  
\begin{equation}
\label{eq:exp_cost_per_arr}
    \mathbb{E}[f(\Delta ({r_n})){\mathds{1}_{\{ \Delta ({r_n}) \le \tau -1 \} }}] = \sum\limits_{t = 1}^{\tau-1}  {f(\Delta (t)){p_n}(t)},
\end{equation}
where by slightly abusing the notation, we let $t$ denote the index of time-slot after  ${u_{k - 1}^\pi}$, and let ${{p_n}(t)}$ be the probability mass function that the $n$-th request arrives at time-slot $t$.
Here ${{p_n}(t)}$ follows the negative binomial distribution \cite{ross2014introduction2}, i.e., 
\begin{equation}
\label{eq:neg_bin_dis}
{p_n}(t) = \left\{ {\begin{array}{*{20}{c}}
{\binom{t-1}{n-1}{\lambda ^n}{{(1 - \lambda )}^{t - n}},\ \ \ t = n,n + 1, \cdots ;}\\
{0,\ \ \ {\rm{otherwise}}}.
\end{array}} \right.
\end{equation}
Plugging Eq.~\eqref{eq:neg_bin_dis} into Eq.~\eqref{eq:exp_cost_per_arr}, we obtain
\begin{equation}
\begin{aligned}
        &\mathbb{E}[f(\Delta ({r_n})){\mathds{1}_{\{ \Delta ({r_n}) \le \tau -1 \} }}] \\ &= \sum\limits_{t = 1}^{\tau-1} {f(\Delta (t))\binom{t-1}{n-1}{\lambda ^n}{{(1 - \lambda )}^{t - n}}},
\end{aligned}
\end{equation}
where the item in the summation equals $0$ when $t < n$.

We rewrite the expected total cost in an update interval as
\begin{equation}
\label{eq:exp_total_cost_one_interval}
    \begin{aligned}
        \mathbb{E}[{C_{k}}] &= p + \sum\limits_{n = 1}^{{N_{k}}} \mathbb{E}[{f(\Delta ({r_n})){\mathds{1}_{\{ {N_{k}} > 0\} }}}] \\
        & = p + \sum\limits_{n = 1}^\infty  {\sum\limits_{t = 1}^{\tau  - 1} {f(\Delta (t))\binom{t-1}{n-1}{\lambda ^n}{{(1 - \lambda )}^{t - n}}} }  \\
        & \buildrel (a) \over = p + \sum\limits_{t = 1}^{\tau  - 1} {\sum\limits_{n = 1}^\infty  {f(\Delta (t))\binom{t-1}{n-1}{\lambda ^n}{{(1 - \lambda )}^{t - n}}} } \\
        & \buildrel (b) \over = p + \sum\limits_{t = 1}^{\tau  - 1} {\sum\limits_{n = 1}^t {f(\Delta (t))\binom{t-1}{n-1}{\lambda ^n}{{(1 - \lambda )}^{t - n}}} }  \\
        & = p + \sum\limits_{t = 1}^{\tau  - 1} {f(\Delta (t))\lambda \sum\limits_{n = 1}^t {\binom{t-1}{n-1}{\lambda ^{n - 1}}{{(1 - \lambda )}^{t - n}}} } \\
        & \buildrel (c) \over = p + \sum\limits_{t = 1}^{\tau-1}  {f(\Delta (t))\lambda {{(\lambda  + (1 - \lambda ))}^{t - 1}}}  \\
         & = p + \sum\limits_{t = 1}^{\tau-1}  {f(\Delta (t))\lambda } \\
        & \buildrel (d) \over  = p + \sum\limits_{t = 1}^{\tau-1}  {f(t)\lambda },
    \end{aligned}
\end{equation}
where we interchange the order of summation in $(a)$ because the sum is finite, (b) is because the maximal number of requests cannot exceed the length of the update interval, (c) comes from the binomial theorem, and (d) is due to the definition of the AoI.
Finally, plugging Eqs.~\eqref{eq:exp_num_requrests} and \eqref{eq:exp_total_cost_one_interval} into Eq.~\eqref{eq:elim} gives Eq.~\eqref{eq:exp_avg_cost}.
\end{proof}  

\bibliographystyle{IEEEtran}
\bibliography{main.bib}

\begin{thebibliography}{10}
\providecommand{\url}[1]{#1}
\csname url@samestyle\endcsname
\providecommand{\newblock}{\relax}
\providecommand{\bibinfo}[2]{#2}
\providecommand{\BIBentrySTDinterwordspacing}{\spaceskip=0pt\relax}
\providecommand{\BIBentryALTinterwordstretchfactor}{4}
\providecommand{\BIBentryALTinterwordspacing}{\spaceskip=\fontdimen2\font plus
\BIBentryALTinterwordstretchfactor\fontdimen3\font minus
  \fontdimen4\font\relax}
\providecommand{\BIBforeignlanguage}[2]{{%
\expandafter\ifx\csname l@#1\endcsname\relax
\typeout{** WARNING: IEEEtran.bst: No hyphenation pattern has been}%
\typeout{** loaded for the language `#1'. Using the pattern for}%
\typeout{** the default language instead.}%
\else
\language=\csname l@#1\endcsname
\fi
#2}}
\providecommand{\BIBdecl}{\relax}
\BIBdecl

\bibitem{zdicccn22}
Z.~Liu, B.~Li, Z.~Zheng, Y.~Hou, and B.~Ji, ``Towards optimal tradeoff between
  data freshness and update cost in information-update systems,'' in
  \emph{Proceedings of ICCCN 2022}, 2022.

\bibitem{Waze}
\BIBentryALTinterwordspacing
Waze mobile app. [Online]. Available: \url{https://www.waze.com/}
\BIBentrySTDinterwordspacing

\bibitem{GasBuddy}
\BIBentryALTinterwordspacing
Gasbuddy mobile app. [Online]. Available: \url{https://www.gasbuddy.com/}
\BIBentrySTDinterwordspacing

\bibitem{li2021achieving}
B.~Li and J.~Liu, ``Achieving information freshness with selfish and rational
  users in mobile crowd-learning,'' \emph{IEEE Journal on Selected Areas in
  Communications}, vol.~39, no.~5, pp. 1266--1276, 2021.

\bibitem{kaul12infocom}
S.~Kaul, R.~Yates, and M.~Gruteser, ``Real-time status: How often should one
  update?'' in \emph{2012 Proceedings IEEE INFOCOM}, 2012, pp. 2731--2735.

\bibitem{li2020waiting}
F.~Li, Y.~Sang, Z.~Liu, B.~Li, H.~Wu, and B.~Ji, ``Waiting but not aging:
  Optimizing information freshness under the pull model,'' \emph{IEEE/ACM
  Transactions on Networking}, vol.~29, no.~1, pp. 465--478, 2021.

\bibitem{zhongdong}
Z.~Liu and B.~Ji, ``Towards the tradeoff between service performance and
  information freshness,'' in \emph{ICC 2019 - 2019 IEEE International
  Conference on Communications (ICC)}, 2019, pp. 1--6.

\bibitem{ling2004tradeoff}
Y.~Ling and J.~Mi, ``An optimal trade-off between content freshness and refresh
  cost,'' \emph{Journal of Applied Probability}, vol.~41, no.~3, p. 721–734,
  2004.

\bibitem{MarketData}
\BIBentryALTinterwordspacing
Market data from morgan stanley. [Online]. Available:
  \url{https://us.etrade.com/l/f/disclosure-library/market-data}
\BIBentrySTDinterwordspacing

\bibitem{sangglobecom17}
Y.~Sang, B.~Li, and B.~Ji, ``The power of waiting for more than one response in
  minimizing the age-of-information,'' in \emph{GLOBECOM 2017-2017 IEEE Global
  Communications Conference}.\hskip 1em plus 0.5em minus 0.4em\relax IEEE,
  2017, pp. 1--6.

\bibitem{huang2015optimizing}
L.~Huang and E.~Modiano, ``Optimizing age-of-information in a multi-class
  queueing system,'' \emph{arXiv preprint arXiv:1504.05103}, 2015.

\bibitem{chen2016age}
K.~Chen and L.~Huang, ``Age-of-information in the presence of error,'' in
  \emph{2016 IEEE International Symposium on Information Theory (ISIT)}, 2016,
  pp. 2579--2583.

\bibitem{kadota2016minimizing}
I.~Kadota, E.~Uysal-Biyikoglu, R.~Singh, and E.~Modiano, ``Minimizing the age
  of information in broadcast wireless networks,'' in \emph{2016 54th Annual
  Allerton Conference on Communication, Control, and Computing (Allerton)},
  2016, pp. 844--851.

\bibitem{costa2014age}
M.~Costa, M.~Codreanu, and A.~Ephremides, ``Age of information with packet
  management,'' in \emph{2014 IEEE International Symposium on Information
  Theory}, 2014, pp. 1583--1587.

\bibitem{najm2016age}
E.~Najm and R.~Nasser, ``Age of information: The gamma awakening,'' in
  \emph{2016 IEEE International Symposium on Information Theory (ISIT)}, 2016,
  pp. 2574--2578.

\bibitem{8619768}
J.~Yun, C.~Joo, and A.~Eryilmaz, ``Optimal real-time monitoring of an
  information source under communication costs,'' in \emph{2018 IEEE Conference
  on Decision and Control (CDC)}, 2018, pp. 4767--4772.

\bibitem{tripathi2021online}
V.~Tripathi and E.~Modiano, ``An online learning approach to optimizing
  time-varying costs of aoi,'' in \emph{Proceedings of the Twenty-second
  International Symposium on Theory, Algorithmic Foundations, and Protocol
  Design for Mobile Networks and Mobile Computing}, 2021, pp. 241--250.

\bibitem{9488746}
K.~Saurav and R.~Vaze, ``Minimizing the sum of age of information and
  transmission cost under stochastic arrival model,'' in \emph{IEEE INFOCOM
  2021 - IEEE Conference on Computer Communications}, 2021, pp. 1--10.

\bibitem{9380899}
R.~D. Yates, Y.~Sun, D.~R. Brown, S.~K. Kaul, E.~Modiano, and S.~Ulukus, ``Age
  of information: An introduction and survey,'' \emph{IEEE Journal on Selected
  Areas in Communications}, vol.~39, no.~5, pp. 1183--1210, 2021.

\bibitem{lu2018age}
N.~Lu, B.~Ji, and B.~Li, ``Age-based scheduling: Improving data freshness for
  wireless real-time traffic,'' in \emph{Proceedings of the eighteenth ACM
  international symposium on mobile ad hoc networking and computing}, 2018, pp.
  191--200.

\bibitem{yates2017age}
R.~D. Yates, P.~Ciblat, A.~Yener, and M.~Wigger, ``Age-optimal constrained
  cache updating,'' in \emph{2017 IEEE International Symposium on Information
  Theory (ISIT)}, 2017, pp. 141--145.

\bibitem{zhong2018two}
J.~Zhong, R.~D. Yates, and E.~Soljanin, ``Two freshness metrics for local cache
  refresh,'' in \emph{2018 IEEE International Symposium on Information Theory
  (ISIT)}, 2018, pp. 1924--1928.

\bibitem{bastopcu2021cache}
M.~Bastopcu and S.~Ulukus, ``Cache freshness in information updating systems,''
  in \emph{2021 55th Annual Conference on Information Sciences and Systems
  (CISS)}, 2021, pp. 01--06.

\bibitem{tang2021cache}
H.~Tang, P.~Ciblat, J.~Wang, M.~Wigger, and R.~D. Yates, ``Cache updating
  strategy minimizing the age of information with time-varying files’
  popularities,'' in \emph{2020 IEEE Information Theory Workshop (ITW)}, 2021,
  pp. 1--5.

\bibitem{8422086}
A.~Arafa, J.~Yang, and S.~Ulukus, ``Age-minimal online policies for energy
  harvesting sensors with random battery recharges,'' in \emph{2018 IEEE
  International Conference on Communications (ICC)}, 2018, pp. 1--6.

\bibitem{wu2018optimal}
X.~Wu, J.~Yang, and J.~Wu, ``Optimal status update for age of information
  minimization with an energy harvesting source,'' \emph{IEEE Transactions on
  Green Communications and Networking}, vol.~2, no.~1, pp. 193--204, 2018.

\bibitem{8437573}
B.~T. Bacinoglu, Y.~Sun, E.~Uysal–Bivikoglu, and V.~Mutlu, ``Achieving the
  age-energy tradeoff with a finite-battery energy harvesting source,'' in
  \emph{2018 IEEE International Symposium on Information Theory (ISIT)}, 2018,
  pp. 876--880.

\bibitem{9080062}
Y.~Dong, P.~Fan, and K.~B. Letaief, ``Energy harvesting powered sensing in iot:
  Timeliness versus distortion,'' \emph{IEEE Internet of Things Journal},
  vol.~7, no.~11, pp. 10\,897--10\,911, 2020.

\bibitem{fountoulakis2020optimal}
E.~Fountoulakis, N.~Pappas, M.~Codreanu, and A.~Ephremides, ``Optimal sampling
  cost in wireless networks with age of information constraints,'' in
  \emph{IEEE INFOCOM 2020-IEEE Conference on Computer Communications Workshops
  (INFOCOM WKSHPS)}.\hskip 1em plus 0.5em minus 0.4em\relax IEEE, 2020, pp.
  918--923.

\bibitem{zhou2019joint}
B.~Zhou and W.~Saad, ``Joint status sampling and updating for minimizing age of
  information in the internet of things,'' \emph{IEEE Transactions on
  Communications}, vol.~67, no.~11, pp. 7468--7482, 2019.

\bibitem{huang2020age}
H.~Huang, D.~Qiao, and M.~C. Gursoy, ``Age-energy tradeoff in fading channels
  with packet-based transmissions,'' in \emph{IEEE INFOCOM 2020-IEEE Conference
  on Computer Communications Workshops (INFOCOM WKSHPS)}.\hskip 1em plus 0.5em
  minus 0.4em\relax IEEE, 2020, pp. 323--328.

\bibitem{tripathi2019age}
V.~Tripathi, R.~Talak, and E.~Modiano, ``Age of information for discrete time
  queues,'' \emph{arXiv preprint arXiv:1901.10463}, 2019.

\bibitem{bertsekas2008introduction}
D.~Bertsekas and J.~N. Tsitsiklis, \emph{Introduction to probability}.\hskip
  1em plus 0.5em minus 0.4em\relax Athena Scientific, 2008, vol.~1.

\bibitem{gallager2013stochastic}
R.~G. Gallager, \emph{Stochastic processes: theory for applications}.\hskip 1em
  plus 0.5em minus 0.4em\relax Cambridge University Press, 2013.

\bibitem{harchol2013performance}
M.~Harchol-Balter, \emph{Performance modeling and design of computer systems:
  queueing theory in action}.\hskip 1em plus 0.5em minus 0.4em\relax Cambridge
  University Press, 2013.

\bibitem{akar2021discrete}
N.~Akar and O.~Dogan, ``Discrete-time queueing model of age of information with
  multiple information sources,'' \emph{IEEE Internet of Things Journal},
  vol.~8, no.~19, pp. 14\,531--14\,542, 2021.

\bibitem{tripathi2017age}
V.~Tripathi and S.~Moharir, ``Age of information in multi-source systems,'' in
  \emph{GLOBECOM 2017-2017 IEEE Global Communications Conference}.\hskip 1em
  plus 0.5em minus 0.4em\relax IEEE, 2017, pp. 1--6.

\bibitem{kosta2021age}
A.~Kosta, N.~Pappas, A.~Ephremides, and V.~Angelakis, ``The age of information
  in a discrete time queue: Stationary distribution and non-linear age mean
  analysis,'' \emph{IEEE Journal on Selected Areas in Communications}, vol.~39,
  no.~5, pp. 1352--1364, 2021.

\bibitem{ross2014introduction}
S.~M. Ross, \emph{Introduction to stochastic dynamic programming}.\hskip 1em
  plus 0.5em minus 0.4em\relax Academic press, 2014.

\bibitem{cacheWorkload-OSDI20}
J.~Yang, Y.~Yue, and K.~V. Rashmi, ``A large scale analysis of hundreds of
  in-memory cache clusters at twitter,'' in \emph{14th USENIX Symposium on
  Operating Systems Design and Implementation (OSDI 20)}.\hskip 1em plus 0.5em
  minus 0.4em\relax USENIX Association, Nov. 2020, pp. 191--208.

\bibitem{royden1988real}
H.~L. Royden and P.~Fitzpatrick, \emph{Real analysis}.\hskip 1em plus 0.5em
  minus 0.4em\relax Macmillan New York, 1988, vol.~32.

\bibitem{ross2014introduction2}
S.~M. Ross, \emph{Introduction to probability models}.\hskip 1em plus 0.5em
  minus 0.4em\relax Academic press, 2014.

\end{thebibliography}

\vspace{-1cm}
\begin{IEEEbiography}[{\includegraphics[width=1in,height=1.25in,clip,keepaspectratio]{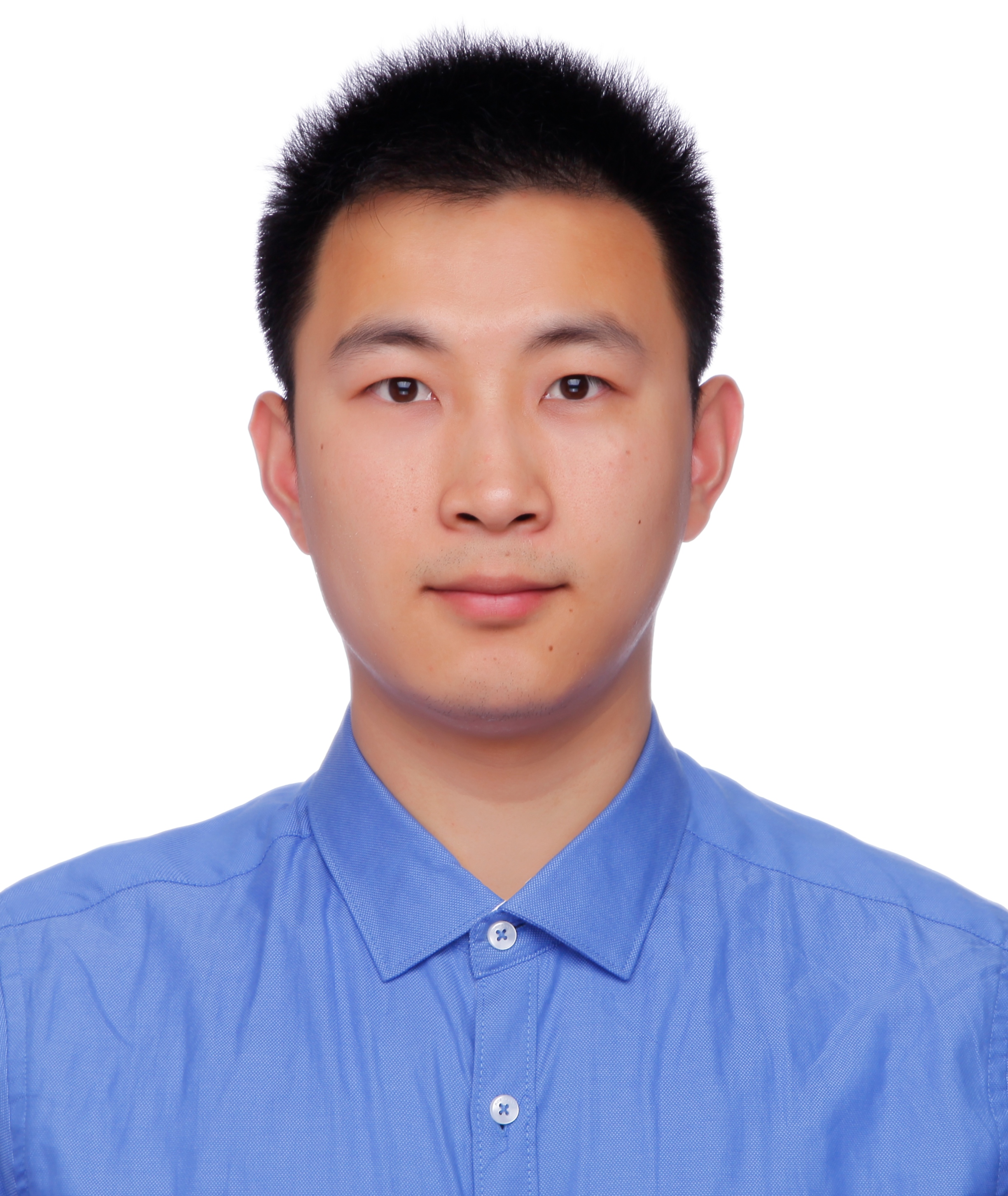}}]{Zhongdong Liu}
is a PhD student in the Department of Computer Science at Virginia Tech. He received his B.S. degree in Mathematics and Applied Mathematics with honor from Northeast Forestry University in 2016. His research interests are in the modeling, analysis, control, and optimization of complex network systems.
\end{IEEEbiography}

\vspace{-0.6cm}
\begin{IEEEbiography}[{\includegraphics[width=1in,height=1.25in,clip,keepaspectratio]{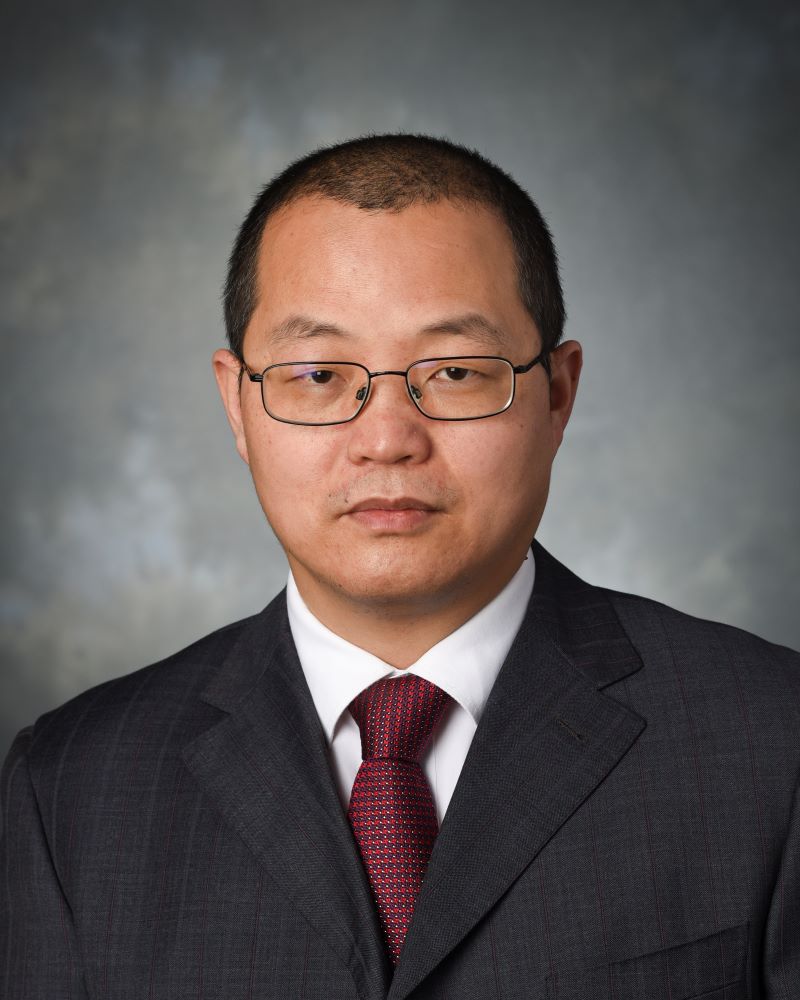}}]{Bin Li}(S'11-M'16-SM'20)
received the B.S. degree in Electronic and Information Engineering, M.S. degree in Communication and Information Engineering, both from Xiamen University, China, and Ph.D. degree in Electrical and  Computer Engineering from The Ohio State University. He is currently an associate professor in the Department of Electrical Engineering at the Pennsylvania State University, University Park, PA, USA. His research focuses on the intersection of networking, machine learning, and system developments, and their applications in networking for virtual/augmented reality, mobile edge computing, mobile crowd-learning, and Internet-of-Things. He is a senior member of the IEEE and a member of the ACM. He received both the National Science Foundation (NSF) CAREER Award and Google Faculty Research Award in 2020, and ACM MobiHoc 2018 Best Poster Award. 
\end{IEEEbiography}

\vspace{-1cm}
\begin{IEEEbiography}[{\includegraphics[width=1in,height=1.25in,clip,keepaspectratio]{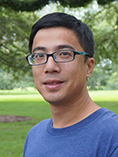}}]{Zizhan Zheng}(S'07-M'10)
received his Ph.D. in Computer Science and Engineering from The Ohio State University in 2010 and his M.S. in Computer Science from Peking University, China, in 2005. He worked as a postdoctoral researcher in the ECE department at The Ohio State University from 2010-2014 and as an associate specialist at UC Davis from 2014-2016. Dr. Zheng joined the CS department of Tulane University as an assistant professor in 2016. His current research interests include networking, trustworthy AI, reinforcement learning, and security. Dr. Zheng is a recipient of the NSF CAREER Award.
\end{IEEEbiography}

\vspace{-1cm}
\begin{IEEEbiography}[{\includegraphics[width=1in,height=1.25in,clip,keepaspectratio]
{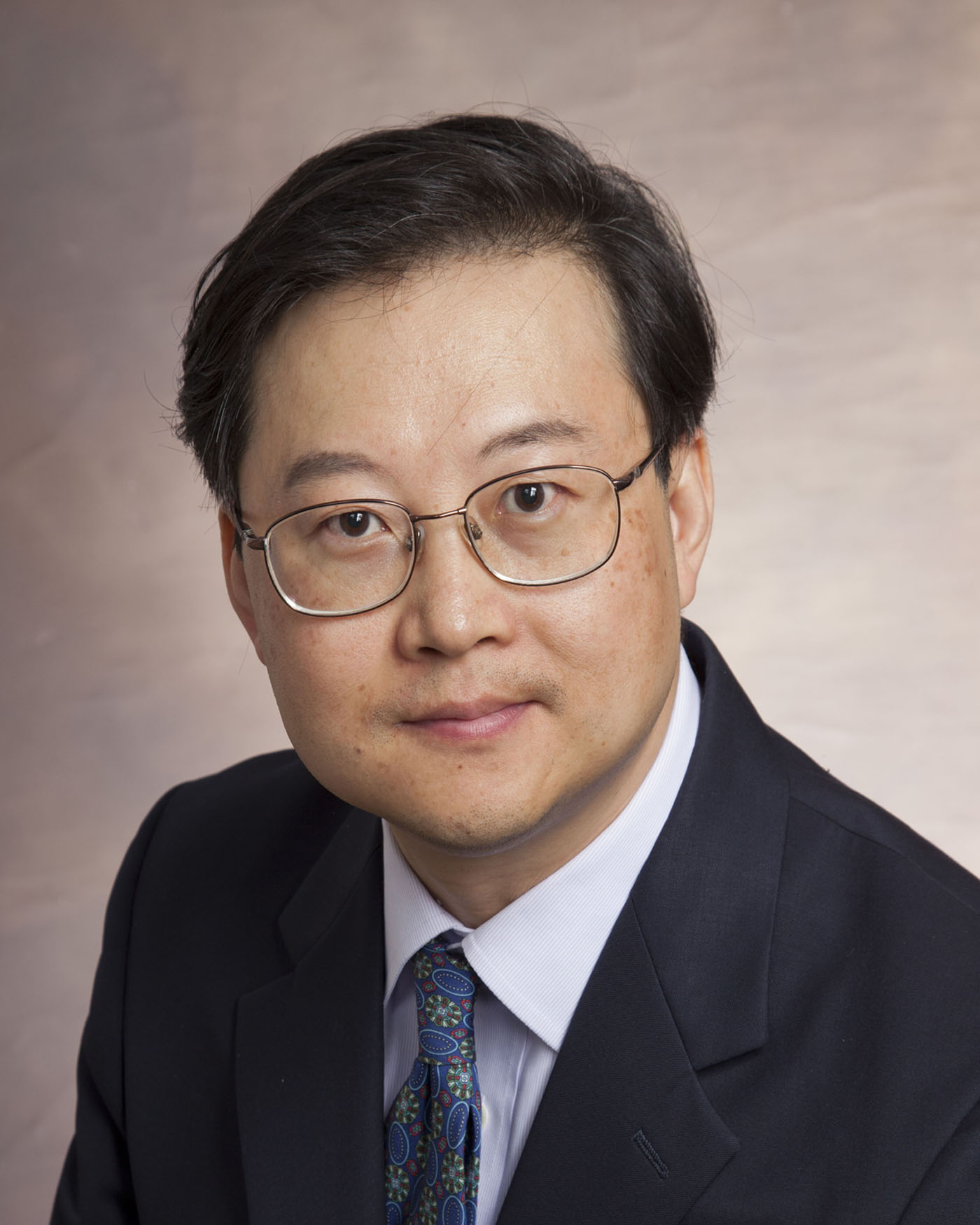}}]{Y. Thomas Hou}(Fellow, IEEE) is Bradley Distinguished Professor of Electrical and Computer Engineering at Virginia Tech, Blacksburg, VA, USA.
He received his Ph.D. degree from NYU Tandon School of Engineering in 1998.
During 1997 to 2002, he was a Member of Research Staff at Fujitsu Laboratories of America, Sunnyvale, CA, USA.  
Prof. Hou’s current research focuses on developing innovative solutions to complex science and engineering problems arising from wireless and mobile networks. 
He is also interested in wireless security.  
He has over 350 papers published in IEEE/ACM journals and conferences. 
His papers were recognized by nine best paper awards from the IEEE and the ACM. 
He holds six U.S. patents.  
He authored/co-authored two graduate textbooks: {\em Applied Optimization Methods for Wireless Networks \/}(Cambridge University Press, 2014) and {\em Cognitive Radio Communications and Networks: Principles and Practices\/} (Academic Press/Elsevier, 2009).
Prof. Hou was named an IEEE Fellow for contributions to modeling and optimization of wireless networks.  
He was/is on the editorial boards of a number of IEEE and ACM transactions and journals. 
He served as Steering Committee Chair of IEEE INFOCOM conference and was a member of the IEEE Communications Society Board of Governors.  He was also a Distinguished Lecturer of the IEEE Communications Society. 
\end{IEEEbiography}

\vspace{-1cm}
\begin{IEEEbiography}[{\includegraphics[width=1in,height=1.25in,clip,keepaspectratio]{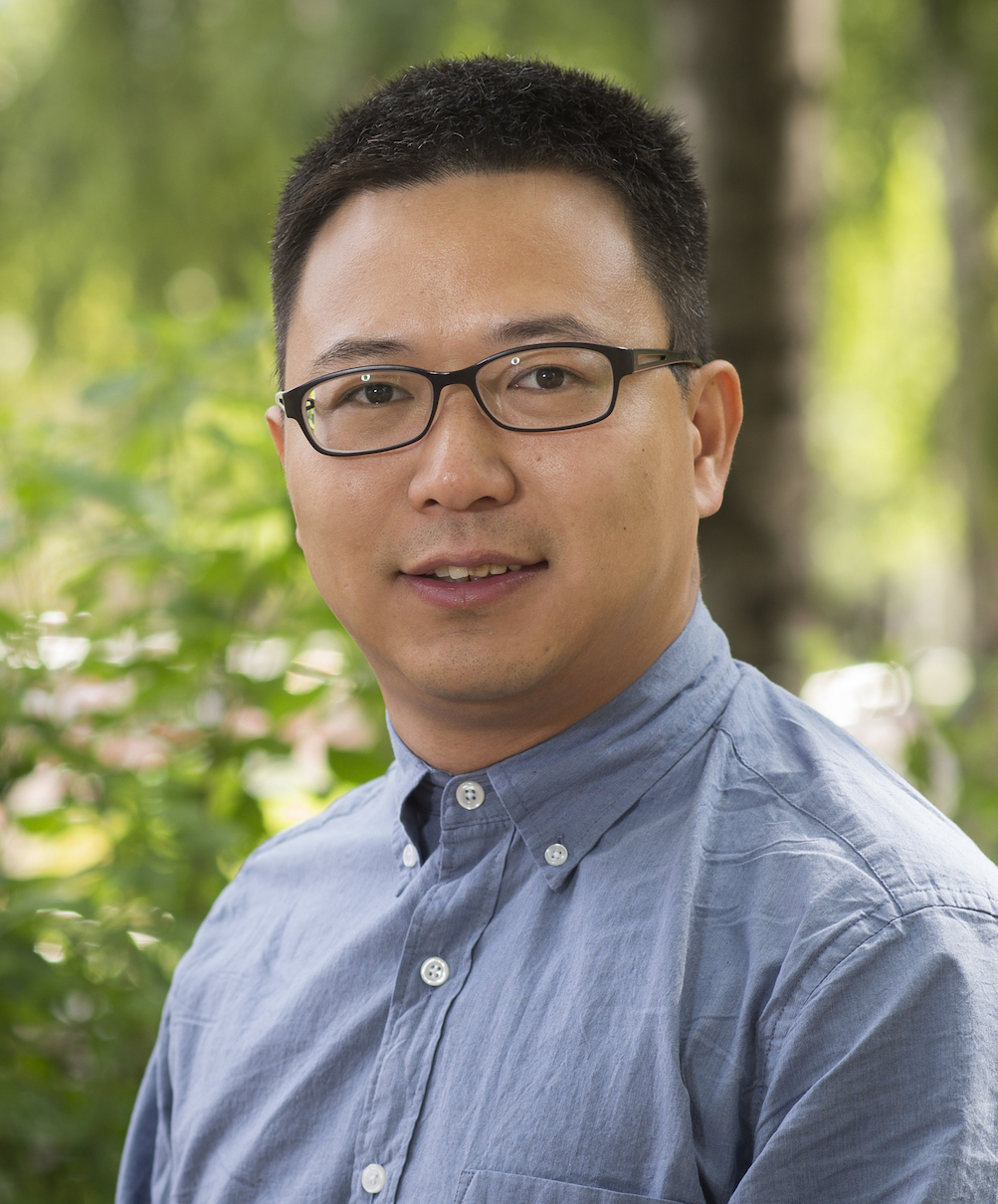}}]{Bo Ji}(S'11-M'12-SM'18)
received his B.E. and M.E. degrees in Information Science and Electronic Engineering from Zhejiang University, Hangzhou, China, in 2004 and 2006, respectively, and his Ph.D. degree in Electrical and Computer Engineering from The Ohio State University, Columbus, OH, USA, in 2012. Dr.~Ji is an Associate Professor in the Department of Computer Science at Virginia Tech, Blacksburg, VA, USA. Prior to joining Virginia Tech, he was an Associate/Assistant Professor in the Department of Computer and Information Sciences at Temple University from July 2014 to July 2020. He was also a Senior Member of the Technical Staff with AT\&T Labs, San Ramon, CA, from January 2013 to June 2014. His research interests are in the modeling, analysis, control, and optimization of computer and network systems, such as wired and wireless networks, large-scale IoT systems, high performance computing systems and data centers, and cyber-physical systems. He has been the general co-chair of IEEE/IFIP WiOpt 2021 and the technical program co-chair of ACM MobiHoc 2023 and ITC 2021, and he has also served on the editorial boards of the IEEE/ACM Transactions on Networking, IEEE Transactions on Network Science and Engineering, IEEE Internet of Things Journal, and IEEE Open Journal of the Communications Society. Dr.~Ji is a senior member of the IEEE and the ACM. He was a recipient of the National Science Foundation (NSF) CAREER Award in 2017, the NSF CISE Research Initiation Initiative Award in 2017, the IEEE INFOCOM 2019 Best Paper Award, the IEEE/IFIP WiOpt 2022 Best Student Paper Award, and the IEEE TNSE Excellent Editor Award in 2021 and 2022.
\end{IEEEbiography}

\end{document}